\documentclass[11pt,english]{article}
\usepackage[T1]{fontenc}
\usepackage[latin9]{inputenc}
\usepackage{geometry}
\usepackage{verbatim}
\usepackage{amsmath}
\usepackage{amsthm}
\usepackage{amssymb}
\usepackage[authoryear]{natbib}

\makeatletter

\providecommand{\tabularnewline}{\\}

\theoremstyle{plain}
\newtheorem{assumption}{\protect\assumptionname}
\theoremstyle{remark}
\newtheorem{rem}{\protect\remarkname}
\theoremstyle{plain}
\newtheorem{thm}{\protect\theoremname}
\theoremstyle{plain}
\newtheorem{cor}{\protect\corollaryname}
\theoremstyle{plain}
\newtheorem{lem}{\protect\lemmaname}

\usepackage{hyperref}
\hypersetup{
     colorlinks   = true,
     citecolor    = blue
}
\date{}

\makeatother

\usepackage{babel}
\providecommand{\assumptionname}{Assumption}
\providecommand{\corollaryname}{Corollary}
\providecommand{\lemmaname}{Lemma}
\providecommand{\remarkname}{Remark}
\providecommand{\theoremname}{Theorem}

\begin{document}
\global\long\def\Acal{\mathcal{A}}%

\global\long\def\ta{\tilde{a}}%

\global\long\def\Bcal{\mathcal{B}}%

\global\long\def\tb{\tilde{b}}%

\global\long\def\Ccal{\mathcal{C}}%

\global\long\def\tc{\tilde{c}}%

\global\long\def\Dcal{\mathcal{D}}%

\global\long\def\hD{\hat{D}}%

\global\long\def\te{\tilde{e}}%

\global\long\def\tF{\tilde{F}}%

\global\long\def\bF{\bar{F}}%

\global\long\def\bG{\bar{G}}%

\global\long\def\hG{\hat{G}}%

\global\long\def\tG{\tilde{G}}%

\global\long\def\Hcal{\mathcal{H}}%

\global\long\def\tH{\tilde{H}}%

\global\long\def\bH{\bar{H}}%

\global\long\def\cH{\check{H}}%

\global\long\def\hL{\hat{L}}%

\global\long\def\Lcal{\mathcal{L}}%

\global\long\def\Mcal{\mathcal{M}}%

\global\long\def\Qcal{\mathcal{Q}}%

\global\long\def\tM{\tilde{M}}%

\global\long\def\bM{\bar{M}}%

\global\long\def\hM{\hat{M}}%

\global\long\def\hr{\hat{r}}%

\global\long\def\tR{\tilde{R}}%

\global\long\def\bR{\bar{R}}%

\global\long\def\RR{\mathbb{R}}%

\global\long\def\SS{\mathbb{S}}%

\global\long\def\ts{\tilde{s}}%

\global\long\def\Tcal{\mathcal{T}}%

\global\long\def\ttt{\tilde{t}}%

\global\long\def\hU{\hat{U}}%

\global\long\def\tu{\tilde{u}}%

\global\long\def\hV{\hat{V}}%

\global\long\def\bW{\bar{W}}%

\global\long\def\tW{\tilde{W}}%

\global\long\def\hw{\hat{w}}%

\global\long\def\tY{\tilde{Y}}%

\global\long\def\tZ{\tilde{Z}}%

\global\long\def\agmin{\arg\min}%

\global\long\def\hbeta{\hat{\beta}}%

\global\long\def\tbeta{\tilde{\beta}}%

\global\long\def\bbeta{\bar{\beta}}%

\global\long\def\teps{\tilde{\varepsilon}}%

\global\long\def\heps{\hat{\varepsilon}}%

\global\long\def\hpi{\hat{\pi}}%

\global\long\def\hxi{\hat{\xi}}%

\global\long\def\txi{\tilde{\xi}}%

\global\long\def\tDelta{\tilde{\Delta}}%

\global\long\def\tdelta{\tilde{\delta}}%

\global\long\def\hdelta{\hat{\delta}}%

\global\long\def\bkappa{\bar{\kappa}}%

\global\long\def\hgamma{\hat{\gamma}}%

\global\long\def\hGamma{\hat{\Gamma}}%

\global\long\def\tGamma{\tilde{\Gamma}}%

\global\long\def\bGamma{\bar{\Gamma}}%

\global\long\def\hSigma{\hat{\Sigma}}%

\global\long\def\bSigma{\bar{\Sigma}}%

\global\long\def\tSigma{\tilde{\Sigma}}%

\global\long\def\hsigma{\hat{\sigma}}%

\global\long\def\hTheta{\hat{\Theta}}%

\global\long\def\dTheta{\dot{\Theta}}%

\global\long\def\bTheta{\bar{\Theta}}%

\global\long\def\cTheta{\check{\Theta}}%

\global\long\def\rTheta{\mathring{\Theta}}%

\global\long\def\hPsi{\hat{\Psi}}%

\global\long\def\hpsi{\hat{\psi}}%

\global\long\def\hPhi{\hat{\Phi}}%

\global\long\def\hphi{\hat{\phi}}%

\global\long\def\tXi{\tilde{\Xi}}%

\global\long\def\lnorm{\left\Vert \right\Vert }%

\global\long\def\abs{\left|\right|}%

\global\long\def\ttheta{\tilde{\theta}}%

\global\long\def\htheta{\hat{\theta}}%

\global\long\def\dv{\dot{v}}%

\global\long\def\rank{{\rm rank}\,}%

\global\long\def\trace{{\rm trace}}%

\global\long\def\Var{{\rm Var}}%

\global\long\def\sign{{\rm sign}}%

\global\long\def\boldone{\mathbf{1}}%

\global\long\def\hmu{\hat{\mu}}%

\global\long\def\hA{\hat{A}}%

\global\long\def\tA{\tilde{A}}%

\global\long\def\hq{\hat{q}}%

\global\long\def\Fcal{\mathcal{F}}%

\global\long\def\tmu{\tilde{\mu}}%

\global\long\def\tX{\tilde{X}}%

\global\long\def\bX{\bar{X}}%

\global\long\def\htau{\hat{\tau}}%

\global\long\def\eig{{\rm eig}}%

\global\long\def\tq{\tilde{q}}%

\global\long\def\ttau{\tilde{\tau}}%

\global\long\def\halpha{\hat{\alpha}}%

\global\long\def\hrho{\hat{\rho}}%

\global\long\def\supp{{\rm supp}}%

\global\long\def\NN{\mathbb{N}}%

\global\long\def\tf{\tilde{f}}%

\global\long\def\tC{\tilde{C}}%

\global\long\def\txi{\tilde{\xi}}%

\global\long\def\tr{\tilde{r}}%

\global\long\def\hzeta{\hat{\zeta}}%

\global\long\def\barm{\bar{m}}%

\global\long\def\texti{\text{(i)}}%

\global\long\def\textii{\text{(ii)}}%

\global\long\def\textiii{\text{(iii)}}%

\global\long\def\oneb{\mathbf{1}}%

\title{Phase transition of the monotonicity assumption in learning local
average treatment effects\thanks{This version: March 24, 2021}}
\author{Yinchu Zhu\thanks{Email: yinchuzhu@brandeis.edu. We thank Tymon S{\l}oczy{\'n}ski,
Kaspar W{\"u}thrich and the participants of MIT Econometrics Lunch
for very helpful discussions and comments. We thank Krisztian Gado
for excellent research assistance. All the errors are mine. }\\
\\
Department of Economics\\
Brandeis University}
\maketitle
\begin{abstract}
We consider the setting in which a strong binary instrument is available
for a binary treatment. The traditional LATE approach assumes the
monotonicity condition stating that there are no defiers (or compliers).
Since this condition is not always obvious, we investigate the sensitivity
and testability of this condition. In particular, we focus on the
question: does a slight violation of monotonicity lead to a small
problem or a big problem? 

We find a phase transition for the monotonicity condition. On one
of the boundary of the phase transition, it is easy to learn the sign
of LATE and on the other side of the boundary, it is impossible to
learn the sign of LATE. Unfortunately, the impossible side of the
phase transition includes data-generating processes under which the
proportion of defiers tends to zero. This boundary of phase transition
is explicitly characterized in the case of binary outcomes. Outside
a special case, it is impossible to test whether the data-generating
process is on the nice side of the boundary. However, in the special
case that the non-compliance is almost one-sided, such a test is possible.
We also provide simple alternatives to monotonicity. %
\end{abstract}

\section{Introduction}

Instrumental variables (IV) regressions have been widely used to study
treatment effects in economics and other disciplines. One important
conceptual framework that justifies the causal interpretation of IV
regressions is the local average treatment (LATE). In this paper,
we consider a simple setting with no covariates and discuss the sensitivity
of the key monotonicity assumption in the LATE framework.

We observe iid data $\{(Y_{i},D_{i},Z_{i})\}_{i=1}^{n}$, where $Y_{i}=Y_{i}(1)D_{i}+Y_{i}(0)(1-D_{i})$,
$D_{i}=D_{i}(1)Z_{i}+D_{i}(0)(1-Z_{i})$ and $Z_{i},D_{i}(1),D_{i}(0)\in\{0,1\}$.
The treatment effect is $Y_{i}(1)-Y_{i}(0)$. (Notice that this assumes
that $Z_{i}$ does not directly affect the potential outcomes: $Y_{i}(z,d)=Y_{i}(d)$
for $z,d\in\{0,1\}$.) Throughout the paper, we maintain the assumption
that the IV is strong ($|cov(D_{i},Z_{i})|\geq C$ for a constant
$C>0$) and the following exogeneity condition
\begin{assumption}
\label{assu: IV exogeneity}$Z_{i}$ is independent of $(Y_{i}(1),Y_{i}(0),D_{i}(1),D_{i}(0))$.
\end{assumption}
The typical IV regression exploits the moment condition $EZ_{i}(Y_{i}-D_{i}\beta)=EZ_{i}E(Y_{i}-D_{i}\beta)$
(i.e., $cov(Z_{i},Y_{i}-D_{i}\beta)=0$). This means\footnote{Under Assumption \ref{assu: IV exogeneity}, $\beta$ can be written
in other ways. For example, $\beta=\frac{E(Y_{i}\mid Z_{i}=1)-E(Y_{i}\mid Z_{i}=0)}{E(D_{i}\mid Z_{i}=1)-E(D_{i}\mid Z_{i}=0)}$.} that
\[
\beta=\frac{E(Y_{i}Z_{i})-E(Y_{i})E(Z_{i})}{E(D_{i}Z_{i})-E(D_{i})E(Z_{i})}.
\]

To describe the causal interpretation of $\beta$, we categorize the
population into four types depending on the value of $(D_{i}(1),D_{i}(0))\in\{0,1\}\times\{0,1\}$.
We introduce their definitions and their probability: 
\begin{equation}
\begin{cases}
P(D_{i}(1)=D_{i}(0)=1)=a & \text{always\ taker}\\
P(D_{i}(1)=1,D_{i}(0)=0)=b & \text{complier}\\
P(D_{i}(1)=0,D_{i}(0)=1)=c & \text{defier}\\
P(D_{i}(1)=0,D_{i}(0)=0)=1-a-b-c & \text{never taker}.
\end{cases}\label{eq: four types}
\end{equation}

As shown in \citet{angrist1996identification}, 
\begin{equation}
\beta=\frac{\mu_{1}b-\mu_{2}c}{b-c},\label{eq: IV beta}
\end{equation}
where $\mu_{1}$ and $\mu_{2}$ are the local average treatment effects
(LATE) for compliers and defiers, respectively: 
\[
\begin{cases}
\mu_{1}=E(Y_{i}(1)-Y_{i}(0)\mid D_{i}(1)=1,D_{i}(0)=0)\\
\mu_{2}=E(Y_{i}(1)-Y_{i}(0)\mid D_{i}(1)=0,D_{i}(0)=1).
\end{cases}
\]

Since (\ref{eq: IV beta}) is in general not a convex combination
of $\mu_{1}$ and $\mu_{2}$, $\beta$ typically does not have a causal
interpretation without further assumptions. The classical assumption
that makes $\beta$ causally interpretable is the following monotonicity
condition. 
\begin{quotation}
\textbf{Monotonicity condition:} either $b=0$ or $c=0$. 
\end{quotation}
Clearly, under the monotonicity condition, (\ref{eq: IV beta}) implies
that $\beta=\mu_{1}$ or $\beta=\mu_{2}$, which can be summarized
as $\beta=E(Y_{i}(1)-Y_{i}(0)\mid D_{i}(1)\neq D_{i}(0))$. Hence,
$\beta$ is interpreted as the average treatment effect on the sub-population
for which $D_{i}(1)\neq D_{i}(0)$. 

As pointed out by \citet{imbens2014instrumental}, perhaps the strongest
justification of the monotonicity condition is when the instrument
provides an incentive to choose the treatment or when the treatment
is simply not an option without $Z_{i}=1$. Outside these situations,
the validity of the monotonicity condition is not always obvious.
In this paper, we try to answer the following questions
\begin{itemize}
\item Suppose that the data can easily reject $H_{0}:\ \beta=0$. If the
monotonicity is slightly violated ($b$ is far from zero but $c$
is close to zero), would this create a big problem or small problem
for learning LATE?
\item In applications with almost one-sided non-compliance ($P(D_{i}=1\mid Z_{i}=0)\approx0$),
should we worry about the interpretation of $\beta$?
\item What are other options for learning LATE without the monotonicity
condition?
\end{itemize}

\subsection{Background of the problem and summary of main results}

Let us explain why (some of) these questions might be quite subtle
and difficult although they seem to have an obvious answer at the
first glance.

The majority of the paper focuses on the seemingly simple question
of learning the sign of LATE (so we can answer the basic question
of whether the treatment is beneficial or harmful). In particular,
whether we can conclude that $\mu_{1}$ and $\beta$ have the same
sign when monotonicity is slightly violated ($c\approx0$). By rearranging
(\ref{eq: IV beta}), we have 
\[
\mu_{1}=\frac{c\mu_{2}+(b-c)\beta}{b}.
\]

Suppose that $\beta<0$. It is easy to see that $\mu_{1}<0$ (and
thus has the same sign as $\beta$) if and only if $\mu_{2}<-(\beta/c)(b-c)$.
Throughout the paper, we assume that $P(|Y_{i}|\leq M)=1$ for a constant
$M>0$. Then the question of learning the sign of $\mu_{1}$ would
seem straight-forward. If $c\rightarrow0$ and $|\beta|$ and $b$
are bounded below by a positive constant, then the threshold $-(\beta/c)(b-c)$
tends to infinity. Since $\mu_{2}$ is bounded (due to the boundedness
of $Y_{i}$), the condition of $\mu_{2}<-(\beta/c)(b-c)$ is asymptotically
satisfied. Hence, the conclusion would be that no matter how slowly
$c$ goes to zero, it is asymptotically valid to conclude that $\mu_{1}$
and $\beta$ have the same sign. 

One subtly is whether modeling $|\beta|$ as a quantity bounded below
by a positive constant is an asymptotic framework that is empirically
relevant. In many empirical studies, if we throw away half of the
data and run the IV regression, we often do not find a statistically
significant $\beta$ anymore. Then it might be too strong to assume
that $|\beta|$ is of a much larger order of magnitude compared to
the estimation noise. Moreover, statistically significance of $\beta$
does not mean that $|\beta|$ is bounded below by a positive constant;
statistical significance is asymptotically guaranteed even if $|\beta|\rightarrow0$
and $\sqrt{n}|\beta|\rightarrow\infty$. To provide robust results
that are empirically relevant, we shall allow $|\beta|\rightarrow0$.
In fact, the use of drifting sequences is the standard practice for
establishing robust analysis in many areas of econometrics.\footnote{Examples include weak instruments (e.g., \citet{staiger1997instrumental}),
local-to-unit-root process (e.g., \citet{Stock1991}), estimation
on the boundary (e.g., \citet{andrews1999estimation}), model selection
(e.g., \citet{Leeb2005}), moment inequalities (e.g., \citet{andrews2009validity})
and time series forecasting (e.g., \citet{hirano2017forecasting})
among others. }

When $|\beta|$ is allowed to tend to zero, the situation is less
straight-forward. When $|b|,c\rightarrow0$,\footnote{Under strong IV condition (say $cov(D_{i},Z_{i})>0$), $b\gtrsim cov(D_{i},Z_{i})$,
which is bounded below by a positive constant.} the threshold of $-(\beta/c)(b-c)$ may or may not be tending to
infinity, depending on the ratio $|\beta|/c$. This paper tries to
find out how worried we should be about $c\rightarrow0$ (but $c\neq0$)
in this case. It turns out that the answer depends on whether $P(D_{i}=1\mid Z_{i}=0)$
is close to zero or not. We now explain our findings. Let us try to
construct a confidence set for the sign of $\mu_{1}$, i.e., a mapping
from the data to a subset of $\{-1,0,1\}$.

The case with $P(D_{i}=1\mid Z_{i}=0)$ being far away from zero is
common, e.g., $P(D_{i}=1\mid Z_{i}=0)>30\%$ in \citet{angrist1998children}.
The question in this case is whether a slight violation of monotonicity
is a big deal. From the discussion above, it is obvious that it is
not a big deal if $|\beta|/c\rightarrow\infty$. The natural way to
proceed is to construct a test or a data-dependent check. If the test
or data check suggests that monotonicity might be a problem, then
use $\{-1,0,1\}$ as the confidence set; if the test suggests otherwise,
then use the more informative set $\{-1\}$ (because $\beta<0$).
This overall procedure has an answer in every situation, regardless
of whether violation of monotonicity is a big deal. However, we show
that if this procedure is robust (i.e., valid with or without monotonicity),
then it must be uninformative (contains both $-1$ and $1$) under
monotonicity ($c=0$). Notice that this is true no matter how sophisticated
the test is. Therefore, although small enough violation of monotonicity
($|\beta|/c\rightarrow\infty$) does not cause a problem, we cannot
really check whether potential violation of monotonicity is small
enough. As a result, if $\beta$ is statistically significant and
the violation of monotonicity tends to zero, this violation may or
may not cause a problem, and we show that no data-dependent procedure
is smart enough to find out (even after imposing constraints such
as $\mu_{1}$ and $\mu_{2}$ having the same sign and both have magnitude
at least $|\beta|$ plus strong distributional restrictions such as
Bernoulli). 

Another common case is $P(D_{i}=1\mid Z_{i}=0)\approx0$. This is
typical when the non-compliance is almost one-sided, e.g., $P(D_{i}=1\mid Z_{i}=0)<2\%$
in the example of Job Training Partnership Act (JPTA). Of course,
$c\rightarrow0$ would still cause a problem if $|\beta|/c\rightarrow0$.
However, since $P(D_{i}=1\mid Z_{i}=0)=a+c\geq c$, we can at least
carve out a ``safe'' region based on the data. For example, if $|\beta|/P(D_{i}=1\mid Z_{i}=0)\rightarrow\infty$,
then $|\beta|/c\rightarrow\infty$ and thus violation of monotonicity
does not cause a problem. Notice that $|\beta|/P(D_{i}=1\mid Z_{i}=0)\rightarrow\infty$
is testable since both $|\beta|$ and $P(D_{i}=1\mid Z_{i}=0)$ can
be learned from the data. In the case of binary outcomes, we provide
a precise characterization of the ``safe'' region. It turns out
that this ``safe'' region also highlights a sharp contrast. If the
data-generating process is in the ``safe'' region, $\mu_{1}$ and
$\beta$ have the same sign; otherwise, the impossibility result from
before holds. 

We refer to this sharp contrast as a phase transition. On side of
the boundary, learning the sign of $\mu_{1}$ is trivial, whereas
it is impossible on the other side of the boundary. There is little
or nothing in the middle. This is the case no matter whether $P(D_{i}=1\mid Z_{i}=0)$
is far away from or close to zero. The difference is that in the former
case, it is impossible to find out on which side of the phase-transition
bound the data-generating process is; the testability is possible
in the latter case. In the former case, we still provide a precise
characterization of the phase transition at least for binary outcomes
because it is useful for robustness checks. For example, suppose that
the boundary of the phase transition is $0.5\%$ of defiers. Although
it is impossible to check which side of the boundary the data-generating
process is, it is still important to know that a mere $1\%$ of defiers
would put the data-generating process on the ``dangerous'' side
of the phase-transition boundary. %

We also outline other ways of learning LATE. We show that the magnitude
of $\mu_{1}$ and $\mu_{2}$ is bounded below by $|\beta|\cdot\gamma$,
where $\gamma$ can be consistently estimated and satisfies $\gamma\asymp|cov(D_{i},Z_{i})|$.
We also show that imposing $|\mu_{1}|\geq|\mu_{2}|$ is enough to
identify the sign of LATE. These results do not rely on monotonicity
at all. %

\subsection{Related literature}

The literature of IV regressions has a long history dating back to
at least \citet{wright1928tariff}. An excellent review on this vast
literature can be found in \citet{imbens2014instrumental}. The framework
of LATE was started by the seminal work of \citet{imbens1994identification},
\citet{angrist1996identification} and \citet{abadie2003semiparametric}.
Since then the LATE-type idea has also been explored in the study
of quantile treatment effects, e.g., \citet{abadie2002instrumental}
and \citet{wuthrich2020comparison}. The framework of LATE fueled
many empirical work ever since the early influential studies including
\citet{angrist1991draft} and \citet{angrist1998children}. The nature
of monotonicity condition has been discussed for decades, e.g., \citet{robins1989analysis},
\citet{balke1995counterfactuals}, \citet{vytlacil2002independence}
and \citet{heckman2005structural}. Since there is not always an obvious
justification for the monotonicity condition, various specification
tests and alternatives have been proposed, see \citet{huber2015testing},
\citet{kitagawa2015test}, \citet{mourifie2017testing}, \citet{de2017tolerating}
and \citet{2011.06695} among many others. Another interesting approach
focuses on the partial identification of average treatment effects
or other quantities under various restrictions, see \citet{balke1997bounds},
\citet{manski2003partial}, \citet{swanson2018partial} and \citet{machado2019instrumental}
among many others. 

\section{\label{sec: sign LATE}Learning the sign of LATE}

In the rest of the paper, we use the following notation. For $x\in\RR$,
let 
\[
\sign(x)=\begin{cases}
1 & \text{if }x>0\\
0 & \text{if }x=0\\
-1 & \text{if }x<0.
\end{cases}
\]

In Table \ref{tab: empirical}, we consider two empirical studies.
In the JPTA study, the treatment $D_{i}$ is job training and $Z_{i}$
is the indicator of the randomized offer of training and the treat.
In the example of \citet{angrist1998children}, we consider case with
$D_{i}$ being the indicator of being more than 2 children and $Z_{i}$
being the indicator of same sex in the first two children. 

\begin{table}

\caption{\label{tab: empirical}Some estimates in two empirical studies}

\begin{centering}
\begin{tabular}{ccccc}
 &  &  &  & \tabularnewline
 & $P(Z_{i}=1)$ & $P(D_{i}=1\mid Z_{i}=1)$ & $P(D_{i}=1\mid Z_{i}=0)$ & \tabularnewline
\cline{1-4} \cline{2-4} \cline{3-4} \cline{4-4} 
JPTA & 0.6662 & 0.6228 & 0.0112 & \tabularnewline
\citet{angrist1998children} & 0.5048 & 0.4105 & 0.3557 & \tabularnewline
 &  &  &  & \tabularnewline
\end{tabular}
\par\end{centering}
\end{table}

We use these two studies to illustrate the two cases. In \citet{angrist1998children},
$P(D_{i}=1\mid Z_{i}=0)$ is not close to zero. Although the interpretation
of the result is clear under the monotonicity condition ($c=0$),
what if we have a small proportion of defiers ($c\approx0$)? We consider
this setting in Section \ref{subsec: small violation of mono}. In
the JPTA study, $P(D_{i}=1\mid Z_{i}=0)$ is close to zero, but does
this mean that we do not need to worry? We provide analysis for this
setting in Section \ref{subsec: small violation of one-sided noncomp}.
We state most of the theoretical results for $\beta<0$, but results
for $\beta>0$ can be obtained analogously. %

\subsection{\label{subsec: small violation of mono}Small violation of monotonicity:
$P(D_{i}=1\mid Z_{i}=0)\gg0$}

For simplicity, we assume that the distribution of $Z_{i}\in\{0,1\}$
is known. We first introduce notations for the distribution of $(Y_{i}(1),Y_{i}(0),D_{i}(1),D_{i}(0))$.
We specify the distribution of $(D_{i}(1),D_{i}(0))$ and then the
conditional distribution of $(Y_{i}(1),Y_{i}(0))\mid(D_{i}(1),D_{i}(0))$.
The former is straight-forward; we simply use the same notation $a,b,c$
as in (\ref{eq: four types}). Define the conditional distribution
\[
H(y_{1},y_{0},d_{1},d_{0})=P\left(Y_{i}(1)\leq y_{1}\ and\ Y_{i}(0)\leq y_{0}\mid D_{i}(1)=d_{1},D_{i}(0)=d_{0}\right).
\]

Let $\theta=(a,b,c,H)$. Then the distribution of $(Z_{i},Y_{i}(1),Y_{i}(0),D_{i}(1),D_{i}(0))$
is indexed by $\theta$. Let $P_{\theta}$ and $E_{\theta}$ denote
the distribution and expectation under $\theta$, respectively. The
following quantities can be written as a function of $\theta$: 
\begin{itemize}
\item LATE for compliers: $\mu_{1}(\theta)=E_{\theta}(Y_{i}(1)-Y_{i}(0)\mid D_{i}(1)=1,D_{i}(0)=0)$
\item LATE for defiers: $\mu_{2}(\theta)=E_{\theta}(Y_{i}(1)-Y_{i}(0)\mid D_{i}(1)=0,D_{i}(0)=1)$
\end{itemize}
From the data $W=\{(Y_{i},D_{i},Z_{i})\}_{i=1}^{n}$, we can identify
the following quantities: 
\begin{itemize}
\item $k_{1}=a+b=E(D_{i}\mid Z_{i}=1)$
\item $k_{2}=a+c=E(D_{i}\mid Z_{i}=0)$
\item $\beta=[\mu_{1}(\theta)b-\mu_{2}(\theta)c]/(b-c)$.
\end{itemize}
Assuming that these three quantities are known, consider the following
parameter space: 
\begin{multline*}
\Theta(\eta)=\biggl\{\theta=(a,b,c,H):\ a,b,c\in[0,1],\ a+b+c\in[0,1],\ a+b=k_{1},\ a+c=k_{2},\\
\max_{d,z\in\{0,1\}}P_{\theta}(|Y_{i}|\geq M\mid D_{i}=d,Z_{i}=z)=0,\ \frac{\mu_{1}(\theta)b-\mu_{2}(\theta)c}{b-c}=\beta,\\
\ |\mu_{1}(\theta)|\geq|\beta|,\ \sign(\mu_{1}(\theta))=\sign(\mu_{2}(\theta)),\ 0\leq c\leq\eta\biggr\},
\end{multline*}
where $M>0$ is a constant. 

Clearly, $\Theta(\eta)$ assumes a lot of structures that are typically
unavailable in practice. In particular, it assumes that $P(D_{i}=1\mid Z_{i}=1)$,
$P(D_{i}=1\mid Z_{i}=0)$ and the population IV regression coefficient
$\beta$ are known. Moreover, it assumes that the LATE for the compliers
and defiers has the same sign and that the magnitude of LATE for compliers
is not too small. The only difficulty is that $c$ (proportion of
defiers) might not be exactly zero and is allowed to be between $0$
and a small tolerance level $\eta$. The point of this subsection
is to show that even under these additional assumptions, allowing
for a small $\eta$ makes it impossible to learn the sign of LATE.
To make this point, we show that many data generating processes with
no defiers and $\sign(\mu_{1})=\beta$ are observationally equivalent
to those with a small proportion of defiers and $\sign(\mu_{1})\neq\beta$. 

To formally state this, we define the following subset
\[
\Theta_{*}=\left\{ \theta=(a,b,c,H)\in\Theta(\eta):\ c=0,\ Q_{1,\theta}(1-\varepsilon_{1})-Q_{2,\theta}(\varepsilon_{1})>\varepsilon_{2}\right\} ,
\]
where $\varepsilon_{1},\varepsilon_{2}>0$ are constants and $Q_{1,\theta}$
and $Q_{2,\theta}$ are the quantile functions of $Y_{i}\mid(D_{i}=1,Z_{i}=0)$
and $Y_{i}\mid(D_{i}=0,Z_{i}=1)$ under $P_{\theta}$, respectively;
in other words, for any $\varepsilon\in(0,1)$,
\[
Q_{1,\theta}(\varepsilon)=\inf\left\{ t\in\RR:\ P_{\theta}\left(Y_{i}\leq t\mid D_{i}=1,Z_{i}=0\right)\geq\varepsilon\right\} 
\]
and 
\[
Q_{2,\theta}(\varepsilon)=\inf\left\{ t\in\RR:\ P_{\theta}\left(Y_{i}\leq t\mid D_{i}=0,Z_{i}=1\right)\geq\varepsilon\right\} .
\]

\begin{rem}
\label{rem: overlap}The condition of $Q_{1,\theta}(1-\varepsilon_{1})-Q_{2,\theta}(\varepsilon_{1})>\varepsilon_{2}$
is not very restrictive. When $Y_{i}$ is binary in $\{0,1\}$, $Q_{1,\theta}(1-\varepsilon_{1})-Q_{2,\theta}(\varepsilon_{1})>\varepsilon_{2}$
holds if 
\[
E_{\theta}(Y_{i}\mid D_{i}=1,Z_{i}=0),E_{\theta}(Y_{i}\mid D_{i}=0,Z_{i}=1)\in(\varepsilon_{1},1-\varepsilon_{1}).
\]

This seems to be reasonable since it might be a bit unrealistic to
expect extreme situations with $E_{\theta}(Y_{i}\mid D_{i}=1,Z_{i}=0)\rightarrow0$
or $E_{\theta}(Y_{i}\mid D_{i}=0,Z_{i}=1)\rightarrow1$. 
\end{rem}
We now state the key observation. 
\begin{thm}
\label{thm: key imps}Let $M,\varepsilon_{2}>0$ and $\varepsilon_{1}\in(0,1)$.
Assume that $\beta<0$, $0<\eta<\varepsilon_{1}\min\{k_{2},1-k_{1},k_{1}-k_{2}\}$
and $3|\beta|/\eta<\varepsilon_{2}/(k_{1}-k_{2})$. Then for any $\theta\in\Theta_{*}$,
there exists $\ttheta\in\Theta(\eta)$ such that \\
(1) $P_{\theta}$ and $P_{\ttheta}$ imply the same distribution for
the observed data $(Y_{i},D_{i},Z_{i})$\\
(2) $\mu_{1}(\theta)=\beta<0$ and $\mu_{1}(\ttheta),\mu_{2}(\ttheta)>-\beta>0$.
\end{thm}
Theorem \ref{thm: key imps} provides the key insight on why lack
of monotonicity creates difficult issues. For a small tolerance level
$\eta$, as long as $|\beta|/\eta$ is not too large, a data-generating
process with no defiers would look exactly like another data-generating
process with a small proportion of defiers such that LATE has different
signs under the two data-generating processes. 

This sheds light on one of the most common problems in IV regressions.
If we reject $H_{0}:\ \beta=0$ and have some arguments against the
presence of defiers (e.g., $Z_{i}$ provides more information and
thus encourages $D_{i}=1$), can we reliably say that $\mu_{1}$ is
non-zero and has the same sign as $\beta$? By Theorem \ref{thm: key imps},
we see that a small proportion of defiers might be enough to invalidate
the result. In the asymptotic framework, rejecting $H_{0}:\ \beta=0$
in large samples is almost guaranteed when $|\beta|\gg n^{-1/2}$.
However, even if $\eta\rightarrow0$ (the proportion of defiers is
small), the observed data is indistinguishable from a distribution
with $\mu_{1}\neq\sign(\beta)$ when $\eta\gg|\beta|$. Therefore,
a slight violation of monotonicity creates a problem if $\eta\gg|\beta|$. 

The natural question is whether or not we could check $\eta\gg|\beta|$
in the data. Unfortunately, the answer is no. We now show this using
an adaptivity argument based on Theorem \ref{thm: key imps}. We define
a confidence set of $\mu_{1}$ to be any measurable function mapping
the observed data $W_{n}$ to a subset of $\{-1,0,1\}$ with a guarantee
on the coverage probability. 
\begin{cor}
\label{cor: adaptivity}Let $M,\varepsilon_{1},\varepsilon_{2},k_{1},k_{2}>0$
be any fixed constants such that $k_{1}-k_{2}>0$. Assume that $\beta<0$,
$\eta\rightarrow0$ and $|\beta|\rightarrow0$ such that $|\beta|\ll\eta$.
Let $CS(W_{n})$ be a confidence set for $\sign(\mu_{1}(\theta))$
with validity over $\Theta(\eta)$, i.e., 
\[
\liminf_{n\rightarrow\infty}\inf_{\theta\in\Theta(\eta)}P_{\theta}\left(\sign(\mu_{1}(\theta))\in CS(W_{n})\right)\geq1-\alpha,
\]
where $\alpha\in(0,1)$. Then 
\[
\liminf_{n\rightarrow\infty}\inf_{\theta\in\Theta_{*}}P_{\theta}\left(\{-1,1\}\subset CS(W_{n})\right)\geq1-2\alpha.
\]
\end{cor}
The assumption of $k_{2}=P(D_{i}=1\mid Z_{i}=0)$ being fixed models
the situation of $P(D_{i}=1\mid Z_{i}=0)$ being far from zero. The
strong IV condition corresponds to the requirement of $k_{1}-k_{2}$
being fixed and positive. There are some important implications of
Corollary \ref{cor: adaptivity}. The following discussions are for
$\eta\rightarrow0$. The same argument obvious holds if $\eta=k_{2}$. 

First, when $n^{-1/2}\ll|\beta|\ll\eta\ll1$, the data allows us to
distinguish $|\beta|$ from zero, but it is still impossible to consistently
estimate the sign of $\mu_{1}$ over $\Theta(\eta)$. To see this,
consider an argument by contradiction. Suppose that there exists a
consistent estimator, a function $\sigma$ that maps $W_{n}$ to $\{-1,1,0\}$
and $\inf_{\theta\in\Theta(\eta)}P_{\theta}(\mu_{1}(\theta)=\sigma(W_{n}))\geq1-o(1)$.
In other words, $\sigma(W_{n})$ is only one value in $\{-1,0,1\}$.
Then we can set $CS(W_{n})=\{\sigma(W_{n})\}$ and the assumption
of Corollary \ref{cor: adaptivity} holds with an arbitrary $\alpha$,
say $\alpha=0.05$. The conclusion of Corollary \ref{cor: adaptivity}
says that with asymptotic probability at least 90\%, $CS(W_{n})=\{\sigma(W_{n})\}$
contains at least two elements, which is impossible since by construction
$\{\sigma(W_{n})\}$ is always a singleton. Hence, no consistent estimator
for $\sign(\mu_{1}(\theta))$ exists on $\Theta(\eta)$ when $n^{-1/2}\ll|\beta|\ll\eta\ll1$.

Second, clever specification tests (for monotonicity or $\eta\gg|\beta|$)
or other data-dependent procedures might not be able to address the
instability arising from a slight violation of the monotonicity condition.
One common purpose of specification tests is to allow us to handle
the problem based on the result of the tests. For example, when the
test tells us the monotonicity fails, we use a cautious set, say $\{-1,0,1\}$,
as the confidence set for $\sign(\mu_{1})$; when the test tells us
that the monotonicity holds, we use $\{\sign(\beta)\}$ as the confidence
set for $\sign(\mu_{1})$. Then by Corollary \ref{cor: adaptivity},
if this confidence set has uniform validity\footnote{One might wonder whether the requirement of uniform validity is too
stringent. It turns out that a similar result holds even if we replace
uniform validity with pointwise validity.} over $\Theta(\eta)$, the confidence set must be uninformative for
$\sign(\mu_{1})$ on the nice set $\Theta_{*}$. If this confidence
set does not have uniform validity over $\Theta(\eta)$, then one
might question why we want to use a specification test in the first
place. Therefore, for the purpose of learning $\sign(\mu_{1})$, even
if we know that $|\mu_{1}(\theta)|\gg n^{-1/2}$, the monotonicity
condition is not really testable even when the alternative is only
a slight violation of monotonicity ($\eta\rightarrow0$).

Third, Corollary \ref{cor: adaptivity} implies a severe lack of adaptivity.
It states that it is impossible to be valid over the bigger set $\Theta(\eta)$
while maintaining efficiency on the nice set $\Theta_{*}$. Hence,
requiring validity over $\Theta(\eta)$ necessarily causes loss of
efficiency on $\Theta_{*}$. Notice that the loss of efficiency is
not on some points in $\Theta_{*}$. The efficiency loss occurs at
\textbf{every} point in $\Theta_{*}$; note that the second inequality
in Corollary \ref{cor: adaptivity} has $\inf_{\theta\in\Theta_{*}}$
rather than $\sup_{\theta\in\Theta_{*}}$. Therefore, the trade-off
of robustness and efficiency is quite stark.

The condition of $\eta\gg|\beta|$ turns out to define the boundary
of a ``phase transition''. We have seen that if we allow for $\eta\gg|\beta|$,
it is impossible to actually learn $\sign(\mu_{1})$. On the other
hand, we can show that if $\eta\ll|\beta|$, learning the sign of
LATE is trivial: $\sign(\mu_{1})=\sign(\beta)$. To see this, notice
that $|\mu_{2}(\theta)|\leq2M$ (since $P_{\theta}(|Y_{i}|\leq M)=1$).
Since $\beta=(\mu_{1}(\theta)b-\mu_{2}(\theta)c)/(b-c)$, it follows
that 
\[
\mu_{1}(\theta)=\lambda\mu_{2}(\theta)+(1-\lambda)\beta\leq2M\lambda+(1-\lambda)\beta,
\]
with $\lambda=c/b$. Since $\lambda=c/(k_{1}-k_{2}+c)$ and $c\leq\eta\ll|\beta|$,
we have that $\lambda=o(|\beta|)$. This means that 
\[
\mu_{1}(\theta)\leq o(|\beta|)+(1-o(|\beta|))\beta=\beta(1+o(1)).
\]

By $\beta<0$, we have $\sign(\mu_{1}(\theta))=\sign(\beta)$ asymptotically.
We now summarize these results.
\begin{thm}[Phase transition]
\label{thm: phase transition}Let $M,\varepsilon_{1},\varepsilon_{2},k_{1},k_{2}>0$
be any fixed constants such that $k_{1}-k_{2}>0$. Assume that $\beta<0$,
$\eta\rightarrow0$ and $|\beta|\rightarrow0$. \\
(1) If $\eta\gg|\beta|$, then for any $CS(W_{n})$ satisfying
\[
\liminf_{n\rightarrow\infty}\inf_{\theta\in\Theta(\eta)}P_{\theta}\left(\sign(\mu_{1}(\theta))\in CS(W_{n})\right)\geq1-\alpha,
\]
with $\alpha\in(0,1)$, we have
\[
\liminf_{n\rightarrow\infty}\inf_{\theta\in\Theta_{*}}P_{\theta}\left(\{-1,1\}\subset CS(W_{n})\right)\geq1-2\alpha.
\]
(2) If $\eta\ll|\beta|$, then 
\[
\liminf_{n\rightarrow\infty}\inf_{\theta\in\Theta(\eta)}P_{\theta}\left(\sign(\mu_{1}(\theta))=\sign(\beta)\right)=1.
\]
\end{thm}
By Theorem \ref{thm: phase transition}, the magnitude of $|\beta|$
serves as the boundary (in rate) of phase transition. A slight violation
may or may not be a huge problem depending on the order of magnitude
of the violation $\eta$. If $|\beta|\ll\eta$, even imposing the
extra condition of $\sign(\mu_{1}(\theta))=\sign(\mu_{2}(\theta))$
does not help with learning the sign of $\mu_{1}$. In contrast, if
$|\beta|\gg\eta$, we can easily learn $\sign(\mu_{1}(\theta))$ without
assuming $\sign(\mu_{1}(\theta))=\sign(\mu_{2}(\theta))$; the proof
of the second part of Theorem \ref{thm: phase transition} does not
rely on this condition. 

Since $\Theta(\eta)$ allows for a large class of distributions, it
is difficult to say much more than the rate. However, when the outcome
variable is binary, we can precisely determine the boundary for the
phase transition.

\subsubsection{\label{subsec: small violation of mono binary}Exact boundary of
the phase transition for binary outcomes}

We define the counterparts of $\Theta(\eta)$ and $\Theta_{*}$ for
the binary outcomes. Let 
\begin{multline}
\Theta_{binary}(\eta)=\biggl\{\theta=(a,b,c,H):\ a,b,c\in[0,1],\ a+b+c\in[0,1],\ a+b=k_{1},\ a+c=k_{2},\\
P_{\theta}(Y_{i}=D_{i}=1\mid Z_{i}=0),P_{\theta}(Y_{i}=D_{i}=0\mid Z_{i}=1)\geq\varepsilon,\\
P_{\theta}(Y_{i}\in\{0,1\})=1,\ 0\leq c\leq\eta\biggr\},\label{eq: binary para space}
\end{multline}
where $\varepsilon>0$ is a constant. The requirement that $P_{\theta}(Y_{i}=D_{i}\mid Z_{i})$
be bounded away from zero and one is mild in many applications. 
\begin{thm}
\label{thm: phase trans binary easy}Let $k_{1},k_{2},\varepsilon\in(0,1)$
be given constants such that $k_{1}-k_{2}>0$. Suppose that $\eta\in[0,k_{2}]$
and $\beta$ satisfy $\beta<0$, $|\beta|\rightarrow0$ and $\eta\rightarrow0$.
\\
(1) If $\eta\geq|\beta|(k_{1}-k_{2})$, then there does not exist
any estimator of $\mu_{1}(\theta)$ that is consistent uniformly over
$\Theta_{binary}(\eta)$.\\
(2) If $\eta<|\beta|(k_{1}-k_{2})$, then $\sign(\mu_{1}(\theta))=\sign(\beta)$
for any $\theta\in\Theta_{binary}(\eta)$.
\end{thm}
When the outcome variable is not binary, we can dichotomize it to
binary variables. We can define the new outcome variable $\tY_{i}(1)=\oneb\{Y_{i}(1)\geq y\}$
and $\tY_{i}(0)=\oneb\{Y_{i}(0)\geq y\}$, where $y$ is given. Then
the treatment effect is how much the treatment changes the probability
of $Y_{i}\geq y$. 

For example, in \citet{angrist1998children}, one outcome variable
of interest $Y_{i}$ is the number of weeks a person worked in a year.
We can set $y=1$ and ask how the treatment changes the probability
of a person working for at least one week. Once we do this, we can
estimate $|\beta|(k_{1}-k_{2})$, the boundary of phase transition.
The results are in Table \ref{tab: AK98}. We see that any tolerance
level of $c$ above $0.52\%$ can cause a serious problem for the
question of whether or not the LATE is negative. Hence, even if the
proportion of defiers is known to be at most 1\%, it might not be
obvious that we can safely conclude a negative LATE.

\begin{table}
\caption{\label{tab: AK98}Estimating the phase-transition boundary using data
in \citet{angrist1998children}}

\begin{centering}
\begin{tabular}{crc}
 &  & \tabularnewline
 & Point estimate & \tabularnewline
\cline{1-2} \cline{2-2} 
$\beta$ & -0.0950 & \tabularnewline
$|\beta|(k_{1}-k_{2})$ & 0.0052 & \tabularnewline
 &  & \tabularnewline
\end{tabular}
\par\end{centering}
{\small{}The outcome variable is whether or not a person worked for
at least one week in the year.}{\small\par}

\end{table}

\subsection{\label{subsec: small violation of one-sided noncomp}What about $P(D_{i}=1\mid Z_{i}=0)\approx0$?}

In many studies, the absence of defiers is justified by the one-sidedness
of non-compliance. For example, $Z_{i}\in\{0,1\}$ is the randomly
assigned treatment and $D_{i}$ is the actually treatment status.
When the compliance is not perfect (i.e., $P(Z_{i}=D_{i})<1$), the
non-compliance is often one-sided: $P(D_{i}=1\mid Z_{i}=1)<1$ but
$P(D_{i}=1\mid Z_{i}=0)=0$. However, we discuss a small violation
to this ideal case $P(D_{i}=1\mid Z_{i}=0)\approx0$, see JPTA in
Table \ref{tab: empirical} as an example.

Here, we provide a discussion for the case of $k_{2}=P(D_{i}=1\mid Z_{i}=0)\rightarrow0$
in the case of binary outcomes. This is different from Theorem \ref{thm: phase trans binary easy},
which assumes that $k_{2}$ is bounded away from zero. Moreover, when
$k_{2}\rightarrow0$, the natural choice of $\eta$ is $\eta=k_{2}$.
To analyze this case, we consider the following parameter space
\[
\Theta_{binary,*}=\biggl\{\theta=(a,b,c,H)\in\Theta_{binary}(k_{2}):\ P_{\theta}(Y_{i}=D_{i}=1\mid Z_{i}=0)<|\beta|(k_{1}-k_{2})\biggr\},
\]
where $\Theta_{binary}(\cdot)$ is defined in (\ref{eq: binary para space}).
\begin{thm}
\label{thm: phase trans binary hard}Let $k_{1},\varepsilon\in(0,1)$
be given constants. Suppose that $k_{2},|\beta|\rightarrow0$ and
$\beta<0$.\\
(1) there does not exist any estimator of $\mu_{1}(\theta)$ that
is consistent over $\Theta_{binary}(k_{2})\backslash\Theta_{binary,*}$.\\
(2) $\sign(\mu_{1}(\theta))=\sign(\beta)$ for any $\theta\in\Theta_{binary,*}$.
\end{thm}
We notice that $P_{\theta}(Y_{i}=D_{i}=1\mid Z_{i}=0)<|\beta|(k_{1}-k_{2})$
(the boundary in Theorem \ref{thm: phase trans binary hard}) is not
the same as applying $\eta=k_{2}$ to Theorem \ref{thm: phase trans binary easy}.
Applying $\eta=k_{2}$ to $\eta<|\beta|(k_{1}-k_{2})$ in Theorem
\ref{thm: phase trans binary easy} leads to $k_{2}<|\beta|(k_{1}-k_{2})$.
However, $P_{\theta}(Y_{i}=D_{i}=1\mid Z_{i}=0)\leq k_{2}$. To see
this, observe that $P_{\theta}(Y_{i}=D_{i}=1\mid Z_{i}=0)=E_{\theta}(Y_{i}D_{i}\mid Z_{i}=0)\leq E_{\theta}(D_{i}\mid Z_{i}=0)=k_{2}$.
What this means in practice is that Theorem \ref{thm: phase trans binary hard}
makes it easier to be on the ``nice'' side of the boundary; instead
of requiring $|\beta|(k_{1}-k_{2})$ to be above $k_{2}$, we require
it to be above $P_{\theta}(Y_{i}=D_{i}=1\mid Z_{i}=0)$. 

A more important implication of Theorem \ref{thm: phase trans binary hard}
is that it is possible to check whether or not we are on the nice
side of the phase-transition boundary. The set $\Theta_{binary,*}$
is defined by the testable condition 
\[
P_{\theta}(Y_{i}=D_{i}=1\mid Z_{i}=0)<|\beta|(k_{1}-k_{2}).
\]

We apply this to the JPTA study. We set the outcome variable to be
$\oneb\{\text{income}<\$50000\}$. This means that we study the effect
of treatment on the probability of earning less than \$50000.\footnote{We choose ``less than'' instead of ``more than'' to get a negative
$\beta$. The interpretation is intuitively the same. Receiving treatment
makes it less likely to earn less than \$50000 so it makes it more
likely to earn at least \$50000.} The results are in Table \ref{tab: JPTA}. Based on only the point
estimates, the data-generating process is in $\Theta_{binary,*}$,
which is the ``nice'' side of the phase-transition boundary. 

\begin{table}
\caption{\label{tab: JPTA}Estimating the phase-transition boundary using data
in JPTA}

\begin{centering}
\begin{tabular}{crc}
 &  & \tabularnewline
 & Point estimate & \tabularnewline
\cline{1-2} \cline{2-2} 
$\beta$ & -0.0363 & \tabularnewline
$|\beta|(k_{1}-k_{2})$ & 0.0222 & \tabularnewline
$P_{\theta}(Y_{i}=1,D_{i}=1\mid Z_{i}=0)$ & 0.0157 & \tabularnewline
\end{tabular}
\par\end{centering}
{\small{}The outcome variable is whether or not a person earns less
than \$50000 a year.}{\small\par}
\end{table}

\section{\label{sec: no monotonicity}Learning LATE without monotonicity:
an alternative}

There are already alternatives to the monotonicity condition in the
literature. We add the following discussions.

\subsection{Learning the magnitude without any additional assumption}

We now outline a lower bound for the magnitude of LATE.
\begin{thm}
\label{thm: key}Let Assumption \ref{assu: IV exogeneity} hold. Assume
$cov(D_{i},Z_{i})\neq0$. Then 
\[
\max\{|\mu_{1}|,|\mu_{2}|\}\geq|\beta|\cdot\gamma,
\]
where 
\[
\gamma=\frac{\left|E(D_{i}\mid Z_{i}=1)-E(D_{i}\mid Z_{i}=0)\right|}{E(D_{i}\mid Z_{i}=1)+E(D_{i}\mid Z_{i}=0)}.
\]
\end{thm}
Notice that $\gamma\gtrsim|cov(D_{i},Z_{i})|$. Therefore, in the
case of strong instruments, the lower bound $|\beta|\cdot\gamma$
is not too small compared to $|\beta|$. From the proof, we can see
that the lower bound is also tight in that the equality can hold (because
the minimum in the proof can be achieved).

It is worth noting that the lower bound in Theorem \ref{thm: key}
can be related to intent-to-treat (ITT) effects. We notice that 
\[
|\beta|\cdot\gamma=\frac{\left|E(Y_{i}\mid Z_{i}=1)-E(Y_{i}\mid Z_{i}=0)\right|}{E(D_{i}\mid Z_{i}=1)+E(D_{i}\mid Z_{i}=0)}=\frac{|ITT|}{E(D_{i}\mid Z_{i}=1)+E(D_{i}\mid Z_{i}=0)}.
\]

Therefore, the lower bound satisfies $|\beta|\cdot\gamma\geq|ITT|/2$.
Therefore, whenever we find that $\beta\neq0$ or $ITT\neq0$, it
means that the treatment effect is not zero and we can use $|\beta|\cdot\gamma$
as a lower bound. 

\subsection{Learning the sign: $|\mu_{1}|\protect\geq|\mu_{2}|$}

Learning the sign of LATE requires extra restrictions. This is inevitable;
otherwise, monotonicity would be testable in Section \ref{subsec: small violation of mono binary}.
It turns out that simple restrictions such as $|\mu_{1}|\geq|\mu_{2}|$
would suffice. 
\begin{thm}
\label{thm: new ID sign}Let Assumption \ref{assu: IV exogeneity}
hold. Suppose that $\beta\neq0$ and $cov(D_{i},Z_{i})>0$. If $|\mu_{1}|\geq|\mu_{2}|$,
then $\sign(\mu_{1})=\sign(\beta)$. 
\end{thm}
We should notice that Theorem \ref{thm: new ID sign} imposes more
than $|\mu_{1}|\geq|\mu_{2}|$. The assumption of $cov(D_{i},Z_{i})>0$
is not without loss of generality since the condition of $|\mu_{2}|\geq|\mu_{1}|$
is not enough to identify the sign of LATE. Therefore, the result
should be viewed in the context of the empirical application. %

\bibliographystyle{apalike}
\bibliography{SC_biblio}

\appendix

\section{Proofs}

\subsection{Proofs of Theorem \ref{thm: key imps} and Corollary \ref{cor: adaptivity}}

We first state an auxiliary result. Since it is a result of elementary
computations, the proof is omitted. 
\begin{lem}
\label{lem: basic calc}Let $\theta=(a,b,c,H)$. Then $P_{\theta}(D_{i}=1\mid Z_{i}=1)=a+b$,
$P_{\theta}(D_{i}=1\mid Z_{i}=0)=a+c$ and 

\[
P_{\theta}(Y_{i}\leq y\mid D_{i}=1,Z_{i}=1)=F_{1,1}(y)\frac{a}{a+b}+F_{1,0}(y)\frac{b}{a+b},
\]
\[
P_{\theta}(Y_{i}\leq y\mid D_{i}=1,Z_{i}=0)=F_{1,1}(y)\frac{a}{a+c}+F_{0,1}(y)\frac{c}{a+c},
\]
\[
P_{\theta}(Y_{i}\leq y\mid D_{i}=0,Z_{i}=1)=G_{0,1}(y)\frac{c}{1-a-b}+G_{0,0}(y)\frac{1-a-b-c}{1-a-b},
\]
\[
P_{\theta}(Y_{i}\leq y\mid D_{i}=0,Z_{i}=0)=G_{1,0}(y)\frac{b}{1-a-c}+G_{0,0}(y)\frac{1-a-b-c}{1-a-c}.
\]

Moreover, $F_{d_{1},d_{0}}(y)=H\left(y,\infty,d_{1},d_{0}\right)=P_{\theta}(Y_{i}(1)\leq y\mid D_{i}(1)=d_{1},D_{i}(0)=d_{0})$
and $G_{d_{1},d_{0}}(y)=H\left(\infty,y,d_{1},d_{0}\right)=P_{\theta}(Y_{i}(0)\leq y\mid D_{i}(1)=d_{1},D_{i}(0)=d_{0})$. 
\end{lem}
\begin{proof}[\textbf{Proof of Theorem \ref{thm: key imps}}]
Fix an arbitrary $\theta=(a,b,c,H)\in\Theta_{*}$. Let $F_{d_{1},d_{0}}$
and $G_{d_{1},d_{0}}$ be the distribution implied by $H$; see Lemma
\ref{lem: basic calc}. By the definition of $\Theta_{*}$, we have
$a=k_{2}$, $b=k_{1}-k_{2}$ and $c=0$. The rest of the proof proceeds
in three steps.

\textbf{Step 1:} define $\ttheta=(\ta,\tb,\tc,\tH)$. 

We choose $\ttheta=(\ta,\tb,\tc,\tH)$ as follows. Let $\tc=\eta$,
$\tb=k_{1}-k_{2}+\tc$ and $\ta=k_{2}-\tc$. We choose $\tH(y_{1},y_{2},d_{1},d_{0})=\tF_{d_{1},d_{0}}(y)\tG_{d_{1},d_{0}}(y)$
with $\tF_{d_{1},d_{0}}$ and $\tG_{d_{1},d_{0}}$ chosen as follows. 

Since $\varepsilon_{2}>3(k_{1}-k_{2})|\beta|/\eta$, there exists
$\delta$ such that $0<\delta<\varepsilon_{2}-3(k_{1}-k_{2})|\beta|/\eta$.
Let $B_{1}=Q_{1,\theta}(1-\varepsilon_{1})-\delta$. By $c=0$, Lemma
\ref{lem: basic calc} implies that $P_{\theta}(Y_{i}\leq y\mid D_{i}=1,Z_{i}=0)=F_{1,1}(y)$.
Since $Q_{1,\theta}$ is the quantile function of $Y_{i}\mid(D_{i}=1,Z_{i}=0)$,
it is the quantile function of $F_{1,1}$. Since $B_{1}<Q_{1,\theta}(1-\varepsilon_{1}$),
we have $F_{1,1}(B_{1})<1-\varepsilon_{1}$. By $\tc/k_{2}=\eta/k_{2}<\varepsilon_{1}$,
we have $F_{1,1}(B_{1})<1-\varepsilon_{1}<1-\tc/k_{2}$. Define 
\[
\tF_{0,1}(y)=\frac{F_{1,1}(y)-F_{1,1}(B_{1})}{1-F_{1,1}(B_{1})}\cdot\oneb\{y\geq B_{1}\}.
\]

Clearly, $\tF_{0,1}$ is a legitimate cumulative distribution function
(cdf) for a random variable with support in $[-M,M]$, , i.e., non-decreasing
and right-continuous.

Let $B_{2}=Q_{2,\theta}(\varepsilon_{1})$. By $c=0$, Lemma \ref{lem: basic calc}
implies that $P_{\theta}(Y_{i}\leq y\mid D_{i}=0,Z_{i}=1)=G_{0,0}(y)$.
Since $Q_{2,\theta}$ is the quantile function of $Y_{i}\mid(D_{i}=0,Z_{i}=1)$,
it is the quantile function of $G_{0,0}$. Since $\tc/(1-k_{1})=\eta/(1-k_{1})<\varepsilon_{1}$,
we have that $G_{0,0}(B_{1})\geq\varepsilon_{1}>\tc/(1-k_{1})$. Define
\[
\tG_{0,1}(y)=\frac{G_{0,0}(y)}{G_{0,0}(B_{2})}\cdot\oneb\{y\leq B_{2}\}.
\]

Again, $\tG_{0,1}$ is a legitimate cdf for a random variable with
support in $[-M,M]$.

We then choose $\tF_{0,0}$ to be the cdf of any random variable with
support in $[-M,M]$, 
\[
\tF_{1,1}(y)=\frac{k_{2}F_{1,1}(y)-\tc\tF_{0,1}(y)}{k_{2}-\tc}
\]
 and 
\[
\tF_{1,0}(y)=\frac{(k_{1}-k_{2})F_{1,0}(y)+\tc\tF_{0,1}(y)}{\tb}.
\]

Finally, we choose $\tG_{1,1}$ to be the cdf of any random variable
with support in $[-M,M]$, $\tG_{0,1}(y)=G_{0,0}(y)$, 
\[
\tG_{0,0}(y)=\frac{(1-k_{1})G_{0,0}(y)-\tc\tG_{0,1}(y)}{1-k_{1}-\tc}
\]
 and
\[
\tG_{1,0}(y)=\frac{(k_{1}-k_{2})G_{1,0}(y)+\tc\tG_{0,1}(y)}{\tb}.
\]

Clearly, $\tF_{1,0}$ and $\tG_{1,0}$ are legitimate cdf's since
each is a convex combination of cdf's (due to $k_{1}-k_{2},\tc>0$
and $\tb=k_{1}-k_{2}+\tc$). We now check $\tF_{1,1}$. By the definition
of $\tF_{0,1}$, we have
\[
\tF_{1,1}(y)=\begin{cases}
\frac{k_{2}}{k_{2}-\tc}F_{1,1}(y) & \text{if }y<B_{1}\\
\frac{1}{k_{2}-\tc}\cdot\left[\left(k_{2}-\frac{\tc}{1-F_{1,1}(B_{1})}\right)\cdot F_{1,1}(y)+\frac{\tc}{1-F_{1,1}(B_{1})}F_{1,1}(B_{1})\right] & \text{if }y\geq B_{1}.
\end{cases}
\]

We first observe $\tF_{1,1}(M)=1$. We also observe that $\tF_{1,1}(B_{1})=\frac{k_{2}}{k_{2}-\tc}F_{1,1}(B_{1})\geq\frac{k_{2}}{k_{2}-\tc}\lim_{y\uparrow B_{1}}F_{1,1}(B_{1})=\lim_{y\uparrow B_{1}}\tF_{1,1}(B_{1})$.
Moreover, $\tF_{1,1}$ is clearly non-decreasing and right-continuous
on $[-M,B_{1})$. It is also right-continuous on $[B_{1},M]$. Since
$F_{1,1}(B_{1})<1-\tc/k_{2}$, we have that $k_{2}-\frac{\tc}{1-F_{1,1}(B_{1})}>0$
and $\tF_{1,1}$ is non-decreasing on $[B_{1},M]$. Therefore, $\tF_{1,1}$
is a legitimate cdf for a distribution supported inside $[-M,M]$. 

Finally, we check $\tG_{0,0}$. We observe 
\[
\tG_{0,0}(y)=\begin{cases}
\frac{1}{1-k_{1}-\tc}\cdot\left(1-k_{1}-\frac{\tc}{G_{0,0}(B_{2})}\right)\cdot G_{0,0}(y) & \text{if }y\leq B_{2}\\
\frac{1-k_{1}}{1-k_{1}-\tc}\cdot G_{0,0}(y) & \text{if }y>B_{2}.
\end{cases}
\]

Since $G_{0,0}(B_{2})\geq\tc/(1-k_{1})$, we have that $1-k_{1}-\frac{\tc}{G_{0,0}(B_{2})}\geq0$
and $\tG_{0,0}$ is non-decreasing on $[-M,B_{2}]$. The rest of the
argument is analogous to that for $\tF_{1,1}$. This concludes that
$\tG_{0,0}$ is a legitimate cdf for a distribution supported inside
$[-M,M]$. Therefore, we have proved that $\{\tF_{d_{1},d_{0}}\}_{d_{1},d_{0}\in\{0,1\}}$
and $\{\tG_{d_{1},d_{0}}\}_{d_{1},d_{0}\in\{0,1\}}$ are cdf's for
a distribution with support in $[-M,M]$.

\textbf{Step 2:} show that $\ttheta\in\Theta(\eta)$. 

We clearly have $\ta+\tb=k_{1}$, $\ta+\tc=k_{2}$, $\tc\in[0,\eta]$,
$\ta,\tb,\tc\in[0,1]$ and $\ta+\tb+\tc\in[0,1]$.

We now show $P_{\ttheta}(|Y_{i}|\leq M\mid D_{i},Z_{i})=1$. By Lemma
\ref{lem: basic calc} (see Step 3), the distribution of $Y_{i}\mid(D_{i}=d,Z_{i}=z)$
is a mixture of $\{\tF_{d_{1},d_{0}}\}_{d_{1},d_{0}\in\{0,1\}}$ and
$\{\tG_{d_{1},d_{0}}\}_{d_{1},d_{0}\in\{0,1\}}$. By Step 1, the support
of $Y_{i}\mid(D_{i}=d,Z_{i}=z)$ is in $[-M,M]$ for any $d,z\in\{0,1\}$. 

It remains to check $[\mu_{1}(\ttheta)\tb-\mu_{2}(\ttheta)\tc]/(\tb-\tc)=\beta$.
We notice that $\mu_{1}(\theta)=\beta$ since $c=0$. Therefore, 
\begin{multline}
\beta=\mu_{1}(\theta)=E_{\theta}(Y_{i}(1)\mid D_{i}(1)=1,D_{i}(0)=0)-E_{\theta}(Y_{i}(0)\mid D_{i}(1)=1,D_{i}(0)=0)\\
=\int ydF_{1,0}(y)-\int ydG_{1,0}(y).\label{eq: main imps 5}
\end{multline}

We now observe 
\begin{multline}
\mu_{2}(\ttheta)=E_{\ttheta}(Y_{i}(1)\mid D_{i}(1)=0,D_{i}(0)=1)-E_{\ttheta}(Y_{i}(0)\mid D_{i}(1)=0,D_{i}(0)=1)\\
=\int yd\tF_{0,1}(y)-\int yd\tG_{0,1}(y).\label{eq: main imps 6}
\end{multline}

Similarly, 
\begin{align}
\mu_{1}(\ttheta) & =E_{\ttheta}(Y_{i}(1)\mid D_{i}(1)=1,D_{i}(0)=0)-E_{\ttheta}(Y_{i}(0)\mid D_{i}(1)=1,D_{i}(0)=0)\nonumber \\
 & =\int yd\tF_{1,0}(y)-\int yd\tG_{1,0}(y)\nonumber \\
 & \overset{\texti}{=}\left(\frac{k_{1}-k_{2}}{\tb}\int ydF_{1,0}(y)+\frac{\tc}{\tb}\int yd\tF_{0,1}(y)\right)-\left(\frac{k_{1}-k_{2}}{\tb}\int ydG_{1,0}(y)+\frac{\tc}{\tb}\int yd\tG_{0,1}(y)\right)\nonumber \\
 & =\frac{k_{1}-k_{2}}{\tb}\left(\int ydF_{1,0}(y)-\int ydG_{1,0}(y)\right)+\frac{\tc}{\tb}\left(\int yd\tF_{0,1}(y)-\int yd\tG_{0,1}(y)\right)\nonumber \\
 & \overset{\textii}{=}\frac{k_{1}-k_{2}}{\tb}\beta+\frac{\tc}{\tb}\mu_{2}(\ttheta),\label{eq: main imps 7}
\end{align}
where (i) follows by the definitions of $\tF_{1,0}$ and $\tG_{1,0}$
and (ii) follows by (\ref{eq: main imps 5}) and (\ref{eq: main imps 6}).

By (\ref{eq: main imps 6}) and (\ref{eq: main imps 7}) as well as
$\tb-\tc=k_{1}-k_{2}$, we have 
\[
\frac{\mu_{1}(\ttheta)\tb-\mu_{2}(\ttheta)\tc}{\tb-\tc}=\frac{\mu_{1}(\ttheta)\tb-\mu_{2}(\ttheta)\tc}{k_{1}-k_{2}}=\frac{(k_{1}-k_{2})\beta}{k_{1}-k_{2}}=\beta.
\]

It remains to show that $\mu_{1}(\ttheta)>-\beta$ and $\mu_{2}(\ttheta)>-\beta$.
First we show $\mu_{1}(\ttheta)>-\beta$. By (\ref{eq: main imps 7}),
we only need to verify $\tc\mu_{2}(\ttheta)>-(k_{1}-k_{2}+\tb)\beta$.
By (\ref{eq: main imps 6}), it suffices to verify
\[
\int yd\tF_{0,1}(y)-\int yd\tG_{0,1}(y)>-\frac{(k_{1}-k_{2}+\tb)\beta}{\tc}=-\frac{[2(k_{1}-k_{2})+\eta]\beta}{\eta}.
\]

Since $\eta<k_{1}-k_{2}$, it is enough to check
\begin{equation}
\int yd\tF_{0,1}(y)-\int yd\tG_{0,1}(y)>-\frac{3(k_{1}-k_{2})\beta}{\eta}.\label{eq: main imps 8}
\end{equation}

Notice that $\tF_{0,1}$ is the cdf of a random variable taking values
in $[B_{1},M]$. Thus, $\int yd\tF_{0,1}(y)\geq B_{1}$. Similarly,
$\tG_{0,1}$ is the cdf of a random variable taking values in $[-M,B_{2}]$,
which means that $\int yd\tG_{0,1}(y)\leq B_{2}$. It follows that
\[
\int yd\tF_{0,1}(y)-\int yd\tG_{0,1}(y)\geq B_{1}-B_{2}=Q_{1,\theta}(1-\varepsilon_{1})-\delta-Q_{2,\theta}(\varepsilon_{1})\geq\varepsilon_{2}-\delta.
\]

By $\delta<\varepsilon_{2}-3(k_{1}-k_{2})|\beta|/\eta$, (\ref{eq: main imps 8})
follows. Hence, we have proved $\mu_{1}(\ttheta)>-\beta$. 

Since $\mu_{2}(\ttheta)=\int yd\tF_{0,1}(y)-\int yd\tG_{0,1}(y)$
and $\eta\leq k_{1}-k_{2}$, (\ref{eq: main imps 8}) implies that
\[
\mu_{2}(\ttheta)>-3\beta>-\beta.
\]

\textbf{Step 3:} show that the observed data $W_{n}$ has the same
distribution under $\theta$ and $\ttheta$.

By Lemma \ref{lem: basic calc}, the distribution of $Y_{i}\mid(D_{i},Z_{i})$
under $\theta$ is given by
\[
P_{\theta}(Y_{i}\leq y\mid D_{i}=1,Z_{i}=1)=F_{1,1}(y)\frac{k_{2}}{k_{1}}+F_{1,0}(y)\frac{k_{1}-k_{2}}{k_{1}},
\]
\[
P_{\theta}(Y_{i}\leq y\mid D_{i}=1,Z_{i}=0)=F_{1,1}(y),
\]
\[
P_{\theta}(Y_{i}\leq y\mid D_{i}=0,Z_{i}=1)=G_{0,0}(y),
\]
\[
P_{\theta}(Y_{i}\leq y\mid D_{i}=0,Z_{i}=0)=G_{1,0}(y)\frac{k_{1}-k_{2}}{1-k_{2}}+G_{0,0}(y)\frac{1-k_{1}}{1-k_{2}}.
\]

Similarly, using $\ta=k_{2}-\tc$ and $\tb=k_{1}-k_{2}+\tc$, the
same calculation in Lemma \ref{lem: basic calc} implies 
\[
P_{\ttheta}(Y_{i}\leq y\mid D_{i}=1,Z_{i}=1)=\tF_{1,1}(y)\frac{\ta}{\ta+\tb}+\tF_{1,0}(y)\frac{\tb}{\ta+\tb}=\tF_{1,1}(y)\frac{k_{2}-\tc}{k_{1}}+\tF_{1,0}(y)\frac{k_{1}-k_{2}+\tc}{k_{1}},
\]
\[
P_{\ttheta}(Y_{i}\leq y\mid D_{i}=1,Z_{i}=0)=\tF_{1,1}(y)\frac{\ta}{\ta+\tc}+\tF_{0,1}(y)\frac{\tc}{\ta+\tc}=\tF_{1,1}(y)\frac{k_{2}-\tc}{k_{2}}+\tF_{0,1}(y)\frac{\tc}{k_{2}},
\]
\begin{multline*}
P_{\ttheta}(Y_{i}\leq y\mid D_{i}=0,Z_{i}=1)=\tG_{0,1}(y)\frac{\tc}{1-\ta-\tb}+\tG_{0,0}(y)\frac{1-\ta-\tb-\tc}{1-\ta-\tb}\\
=\tG_{0,1}(y)\frac{\tc}{1-k_{1}}+\tG_{0,0}(y)\frac{1-k_{1}-\tc}{1-k_{1}},
\end{multline*}
\begin{multline*}
P_{\ttheta}(Y_{i}\leq y\mid D_{i}=0,Z_{i}=0)=\tG_{1,0}(y)\frac{\tb}{1-\ta-\tc}+\tG_{0,0}(y)\frac{1-\ta-\tb-\tc}{1-\ta-\tc}\\
=\tG_{1,0}(y)\frac{k_{1}-k_{2}+\tc}{1-k_{2}}+\tG_{0,0}(y)\frac{1-k_{1}-\tc}{1-k_{2}}.
\end{multline*}

By the definitions of $\tF_{d_{1},d_{0}}$and $\tG_{d_{1},d_{0}}$,
we can easily check that the two sets of equations match in all the
four relations, thereby concluding that $P_{\theta}(Y_{i}\leq y\mid D_{i},Z_{i})$
and $P_{\ttheta}(Y_{i}\leq y\mid D_{i},Z_{i})$ are the same distribution.
Under $\theta=(a,b,c,H)$, $D_{i}\mid Z_{i}$ depends only on $(a+b,a+c)$.
Since $a+b=\ta+\tb=k_{1}$ and $a+c=\ta+\tc=k_{2}$, we have that
$D_{i}\mid Z_{i}$ has the same distribution under $P_{\theta}$ and
$P_{\ttheta}$. Therefore, the observed data has the distribution
under $P_{\theta}$ and $P_{\ttheta}$. The proof is complete. 
\end{proof}
\begin{proof}[\textbf{Proof of Corollary \ref{cor: adaptivity}}]
Throughout the proof, we assume that $n$ is large enough so we have
$0<\eta<\varepsilon_{1}\min\{k_{2},1-k_{1},k_{1}-k_{2}\}$ and $|\beta|/\eta<\varepsilon_{2}/(k_{1}-k_{2})$.
Fix an arbitrary $\theta_{0}\in\Theta_{*}$. Notice that 
\begin{align}
P_{\theta_{0}}\left(\{-1,1\}\subset CS(W_{n})\right) & =P_{\theta_{0}}\left(\{-1\in CS(W_{n})\}\bigcap\{1\in CS(W_{n})\}\right)\nonumber \\
 & =1-P_{\theta_{0}}\left(\{-1\notin CS(W_{n})\}\bigcup\{1\notin CS(W_{n})\}\right)\nonumber \\
 & \geq1-P_{\theta_{0}}(-1\notin CS(W_{n}))-P_{\theta_{0}}(1\notin CS(W_{n}))\nonumber \\
 & =P_{\theta_{0}}(-1\in CS(W_{n}))+P_{\theta_{0}}(1\in CS(W_{n}))-1.\label{eq: cor adaptivity 3}
\end{align}

By Theorem \ref{thm: key imps}, there exists $\ttheta\in\Theta(\eta)$
such that $\mu_{1}(\ttheta)>0$ and $W_{n}$ has the same distribution
under $P_{\theta_{0}}$ and $P_{\ttheta}$. This means that 
\[
P_{\theta_{0}}\left(1\in CS(W_{n})\right)=P_{\ttheta}(1\in CS(W_{n}))=P_{\ttheta}(\sign(\mu_{1}(\ttheta))\in CS(W_{n})).
\]

Moreover, since $\mu_{1}(\theta_{0})=\beta<0$, we have 
\[
P_{\theta_{0}}(-1\in CS(W_{n}))=P_{\theta_{0}}(\sign(\mu_{1}(\theta_{0}))\in CS(W_{n})).
\]

Therefore, 
\begin{align*}
P_{\theta_{0}}\left(\{-1,1\}\subset CS(W_{n})\right) & \geq P_{\theta_{0}}(\sign(\mu_{1}(\theta_{0}))\in CS(W_{n}))+P_{\ttheta}(\sign(\mu_{1}(\ttheta))\in CS(W_{n}))-1\\
 & \overset{\texti}{\geq}2\cdot\inf_{\theta\in\Theta(\eta)}P_{\theta}(\sign(\mu_{1}(\theta))\in CS(W_{n}))-1,
\end{align*}
where (i) follows by $\theta_{0},\ttheta\in\Theta(\eta)$. Since the
above bound holds for an arbitrary $\theta_{0}\in\Theta_{*}$ and
the right-hand side does not depend on $\theta_{0}$, we can take
an infimum over $\theta_{0}$, obtaining 
\[
\inf_{\theta_{0}\in\Theta_{*}}P_{\theta_{0}}\left(\{-1,1\}\subset CS(W_{n})\right)\geq2\cdot\inf_{\theta\in\Theta(\eta)}P_{\theta}(\sign(\mu_{1}(\theta))\in CS(W_{n}))-1.
\]

Now we take $\liminf$ on both sides and use $\liminf_{n\rightarrow\infty}\inf_{\theta\in\Theta(\eta)}P_{\theta}\left(\sign(\mu_{1}(\theta))\in CS(W_{n})\right)\geq1-\alpha$.
The desired result follows. 
\end{proof}

\subsection{Proof of Theorem \ref{thm: phase trans binary easy}}

We start with two auxiliary results. 
\begin{lem}
\label{lem: binary lem part 1}Let $\theta=(a,b,c,H)$ satisfy $P_{\theta}(Y_{i}\in\{0,1\})=1$.
Assume that $P_{\theta}(Y_{i}=D_{i}=1\mid Z_{i}=0)$, $P_{\theta}(Y_{i}=D_{i}=0\mid Z_{i}=1)$,
$k_{2}$ and $k_{1}-k_{2}$ are bounded below by a positive constant.
If $\eta\in[0,k_{2}]$ and $\beta$ satisfy $\beta<0$, $c=0$, $|\beta|\rightarrow0$,
$\eta\rightarrow0$ and $\beta(k_{1}-k_{2})+\eta\geq0$, then for
large enough $n$ there exists $\ttheta=(\ta,\tb,\tc,\tH)$ such that
(1) $P_{\theta}$ and $P_{\ttheta}$ imply the same distribution for
the observed data $(Y_{i},D_{i},Z_{i})$ and (2) $\mu_{1}(\ttheta)\geq0$
and $\tc\leq\eta$.
\end{lem}
\begin{proof}
We follow a similar argument as in the proof of Theorem \ref{thm: key imps}.
Since $Y_{i}\in\{0,1\}$, we can simplify $H$ to 8 numbers as follows.
We parametrize $Y_{i}(d)\mid(D_{i}(1),D_{i}(0))$ as $E(Y_{i}(1)\mid D_{i}(1)=d_{1},D_{i}(0)=d_{0})=r_{d_{1},d_{0}}$
and $E(Y_{i}(0)\mid D_{i}(1)=d_{1},D_{i}(0)=d_{0})=t_{d_{1},d_{0}}$,
where $d_{1},d_{0}\in\{0,1\}$. 

Fix $\theta=(a,b,c,r,t)$ satisfying the following: 
\begin{itemize}
\item $c=0$. 
\item $\beta=r_{1,0}-t_{1,0}<0$.
\item $k_{1}=P_{\theta}(D_{i}=1\mid Z_{i}=1)=a+b=k_{2}+b$
\item $k_{2}=P_{\theta}(D_{i}=1\mid Z_{i}=0)=a+c=a$.
\end{itemize}
Similar to Lemma \ref{lem: basic calc}, we observe that
\[
\rho_{1,1}:=E_{\theta}(Y_{i}\mid D_{i}=1,Z_{i}=1)=r_{1,1}\frac{k_{2}}{k_{1}}+r_{1,0}\frac{k_{1}-k_{2}}{k_{1}},
\]
\[
\rho_{1,0}:=E_{\theta}(Y_{i}\mid D_{i}=1,Z_{i}=0)=r_{1,1},
\]
\[
\rho_{0,1}:=E_{\theta}(Y_{i}\mid D_{i}=0,Z_{i}=1)=t_{0,0},
\]
\[
\rho_{0,0}:=E_{\theta}(Y_{i}\mid D_{i}=0,Z_{i}=0)=t_{1,0}\frac{k_{1}-k_{2}}{1-k_{2}}+t_{0,0}\frac{1-k_{1}}{1-k_{2}}.
\]

Since $\beta=r_{1,0}-t_{1,0}$, we can write $\beta=(\rho_{1,1}k_{1}-\rho_{1,0}k_{2}-\rho_{0,0}(1-k_{2})+\rho_{0,1}(1-k_{1}))/(k_{1}-k_{2})$.
This means that we can eliminate $\rho_{0,0}$ from future calculations
by observing
\[
\rho_{0,0}=\frac{\rho_{1,1}k_{1}-\rho_{1,0}k_{2}+\rho_{0,1}(1-k_{1})-\beta(k_{1}-k_{2})}{1-k_{2}}.
\]

We observe that 
\begin{multline}
\rho_{1,0}k_{2}=P_{\theta}(Y_{i}=1\mid D_{i}=1,Z_{i}=0)\cdot P_{\theta}(D_{i}=1\mid Z_{i}=0)\\
=P_{\theta}(Y_{i}=1,D_{i}=1\mid Z_{i}=0)\label{eq: binary lem part 1 eq 2}
\end{multline}
and similarly
\begin{multline}
(1-t_{0,0})(1-k_{1})=P_{\theta}(Y_{i}=0\mid D_{i}=0,Z_{i}=1)\cdot P_{\theta}(D_{i}=0\mid Z_{i}=1)\\
=P_{\theta}(Y_{i}=0,D_{i}=0\mid Z_{i}=1).\label{eq: binary lem part 1 eq 2.5}
\end{multline}

We now construct $\ttheta=(\ta,\tb,\tc,\tr,\ttt)$. We set $\tb=k_{1}-k_{2}+\tc$
and $\ta=k_{2}-\tc$, where $\tc=\min\{k_{2},\eta\}$. Similar to
Lemma \ref{lem: basic calc}, we observe

\[
E_{\ttheta}(Y_{i}\mid D_{i}=1,Z_{i}=1)=\tr_{1,1}\frac{\ta}{\ta+\tb}+\tr_{1,0}\frac{\tb}{\ta+\tb}=\tr_{1,1}\frac{k_{2}-\tc}{k_{1}}+\tr_{1,0}\frac{k_{1}-k_{2}+\tc}{k_{1}},
\]
\[
E_{\ttheta}(Y_{i}\mid D_{i}=1,Z_{i}=0)=\tr_{1,1}\frac{\ta}{\ta+\tc}+\tr_{0,1}\frac{\tc}{\ta+\tc}=\tr_{1,1}\frac{k_{2}-\tc}{k_{2}}+\tr_{0,1}\frac{\tc}{k_{2}},
\]
\[
E_{\ttheta}(Y_{i}\mid D_{i}=0,Z_{i}=1)=\ttt_{0,1}\frac{\tc}{1-\ta-\tb}+\ttt_{0,0}\frac{1-\ta-\tb-\tc}{1-\ta-\tb}=\ttt_{0,1}\frac{\tc}{1-k_{1}}+\ttt_{0,0}\frac{1-k_{1}-\tc}{1-k_{1}},
\]
\[
E_{\ttheta}(Y_{i}\mid D_{i}=0,Z_{i}=0)=\ttt_{1,0}\frac{\tb}{1-\ta-\tc}+\ttt_{0,0}\frac{1-\ta-\tb-\tc}{1-\ta-\tc}=\ttt_{1,0}\frac{k_{1}-k_{2}+\tc}{1-k_{2}}+\ttt_{0,0}\frac{1-k_{1}-\tc}{1-k_{2}}.
\]

As in the proof of Theorem \ref{thm: key imps}, the distribution
of $(D_{i},Z_{i})$ is the same under $P_{\theta}$ and under $P_{\ttheta}$.
The distribution of $Y_{i}\mid(D_{i},Z_{i})$ is also identical under
$P_{\theta}$ and $P_{\ttheta}$ if the conditional mean $E(Y_{i}\mid D_{i},Z_{i})$
matches in all the four cases of $(D_{i},Z_{i})\in\{0,1\}\times\{0,1\}$.
This means that
\[
r_{1,1}k_{2}+r_{1,0}(k_{1}-k_{2})=\tr_{1,1}(k_{2}-\tc)+\tr_{1,0}(k_{1}-k_{2}+\tc)
\]
\[
r_{1,1}k_{2}=\tr_{1,1}(k_{2}-\tc)+\tr_{0,1}\tc
\]
\[
t_{0,0}(1-k_{2})=\ttt_{0,1}\tc+\ttt_{0,0}(1-k_{2}-\tc)
\]
\[
t_{1,0}(k_{1}-k_{2})+t_{0,0}(1-k_{1})=\ttt_{1,0}(k_{1}-k_{2}+\tc)+\ttt_{0,0}(1-k_{1}-\tc).
\]

Once we impose the constraints of $\tr_{d_{1},d_{0}},\ttt_{d_{1},d_{0}}\in[0,1]$,
we have that
\[
\tr_{1,1}=\frac{r_{1,1}k_{2}-\tr_{0,1}\tc}{k_{2}-\tc}
\]
\[
\tr_{1,0}=\frac{r_{1,0}(k_{1}-k_{2})+\tr_{0,1}\tc}{k_{1}-k_{2}+\tc}
\]
\[
\ttt_{0,0}=\frac{t_{0,0}(1-k_{2})-\ttt_{0,1}\tc}{1-k_{2}-\tc}
\]
\[
\ttt_{1,0}=\frac{t_{1,0}(k_{1}-k_{2})+\ttt_{0,1}\tc}{k_{1}-k_{2}+\tc},
\]
as well as 
\begin{equation}
\max\left\{ 0,\ 1+\frac{(\rho_{1,0}-1)k_{2}}{\tc}\right\} \leq\tr_{0,1}\leq\min\left\{ 1,\ \frac{\rho_{1,0}k_{2}}{\tc}\right\} \label{eq: binary lem part 1 eq 4}
\end{equation}
and 
\begin{equation}
\max\left\{ 0,\ 1+\frac{(t_{0,0}-1)(1-k_{2})}{\tc}\right\} \leq\ttt_{0,1}\leq\min\left\{ 1,\ \frac{t_{0,0}(1-k_{2})}{\tc}\right\} .\label{eq: binary lem part 1 eq 5}
\end{equation}

Since $\tc=o(1)$, we have $k_{2}+\tc\leq1$. By this and $\tc\leq k_{2}$,
the above inequalities can hold. We set $\tr_{0,1}=\min\{1,\rho_{1,0}k_{2}/\tc\}$
and $\ttt_{0,1}=\max\{0,1+(t_{0,0}-1)(1-k_{1})/\tc\}$. Then 
\begin{align*}
\mu_{1}(\ttheta)=\tr_{1,0}-\ttt_{1,0} & =\frac{r_{1,0}(k_{1}-k_{2})+\tr_{0,1}\tc}{k_{1}-k_{2}+\tc}-\frac{t_{1,0}(k_{1}-k_{2})+\ttt_{0,1}\tc}{k_{1}-k_{2}+\tc}\\
 & =\frac{(r_{1,0}-t_{1,0})(k_{1}-k_{2})+(\tr_{0,1}-\ttt_{0,1})\tc}{k_{1}-k_{2}+\tc}=\frac{\beta(k_{1}-k_{2})+(\tr_{0,1}-\ttt_{0,1})\tc}{k_{1}-k_{2}+\tc}.
\end{align*}

It only remains to show that 
\[
\beta(k_{1}-k_{2})+(\tr_{0,1}-\ttt_{0,1})\tc\geq0.
\]

By the definitions of $\tr_{0,1}$ and $\ttt_{0,1}$, we need to show
that 
\begin{equation}
\beta(k_{1}-k_{2})+\min\{\tc,\rho_{1,0}k_{2}\}-\max\{0,\tc+(t_{0,0}-1)(1-k_{1})\}\geq0.\label{eq: binary lem part 1 eq 8}
\end{equation}

By (\ref{eq: binary lem part 1 eq 2}) and (\ref{eq: binary lem part 1 eq 2.5}),
the assumptions imply that $\rho_{1,0}k_{2}$ and $(1-t_{0,0})(1-k_{1})$
are bounded below by a positive constant. For large $n$, $\tc=\min\{k_{2},\eta\}=\eta\rightarrow0$.
Then for large enough $n$, $\min\{\tc,\rho_{1,0}k_{2}\}=\eta$ and
$\max\{0,\tc+(t_{0,0}-1)(1-k_{1})\}=0$. Thus, (\ref{eq: binary lem part 1 eq 8})
becomes $\beta(k_{1}-k_{2})+\eta\geq0$, which is assumed to be true.
The proof is complete. 
\end{proof}
\begin{lem}
\label{lem: binary lem part 2}Let $\theta=(a,b,c,H)$ satisfy $P_{\theta}(Y_{i}\in\{0,1\})=1$.
Assume that $P_{\theta}(Y_{i}=D_{i}=1\mid Z_{i}=0)$, $P_{\theta}(Y_{i}=D_{i}=0\mid Z_{i}=1)$,
$k_{2}$ and $k_{1}-k_{2}$ are bounded below by a positive constant.
If $\eta\in[0,k_{2}]$ and $\beta$ satisfy $\beta<0$, $c\leq\eta$,
$|\beta|\rightarrow0$, $\eta\rightarrow0$ and $\beta(k_{1}-k_{2})+\eta<0$,
then $\mu_{1}(\theta)<0$.
\end{lem}
\begin{proof}
Recall that $\beta=\frac{\mu_{1}(\theta)\cdot b-\mu_{2}(\theta)\cdot c}{b-c}$.
Therefore,
\[
\mu_{1}(\theta)=\frac{c}{b}\mu_{2}(\theta)+\left(1-\frac{c}{b}\right)\beta.
\]

Since $\beta<0$, it suffices to show that $c\mu_{2}(\theta)+(b-c)\beta<0$.
Since $b-c=k_{1}-k_{2}$, we need to show that $c\mu_{2}(\theta)<-\beta(k_{1}-k_{2})$.
Since $c\leq\eta$ and $\mu_{2}(\theta)\in\{0,1\}$, we have $c\mu_{2}\leq\eta$.
Hence, $c\mu_{2}(\theta)<-\beta(k_{1}-k_{2})$ by the assumption of
$\beta(k_{1}-k_{2})+\eta<0$. The proof is complete. 
\end{proof}
\begin{proof}[\textbf{Proof of Theorem \ref{thm: phase trans binary easy}}]
Part (2) follows by Lemma \ref{lem: binary lem part 2}. 

For part (1), we apply Lemma \ref{lem: binary lem part 1} and follow
the same argument as the proof of Corollary \ref{cor: adaptivity}.
Therefore, when $\beta(k_{1}-k_{2})+\eta\geq0$, if a confidence set
$CS(W_{n})$ satisfies
\[
\liminf_{n\rightarrow\infty}\inf_{\theta\in\Theta_{binary}(\eta)}P_{\theta}(\mu_{1}(\theta)\in CS(W_{n}))\geq1-\alpha
\]
for $\alpha\in(0,1)$, then 
\[
\liminf_{n\rightarrow\infty}\inf_{\theta\in\Theta_{binary}(0)}P_{\theta}(\{-1,0\}\subset CS(W_{n}))\geq1-\alpha.
\]

Suppose that a consistent estimator exists, i.e., $\rho(W_{n})\in\{-1,0,1\}$
and $\liminf_{n\rightarrow\infty}\inf_{\theta\in\Theta_{binary}(\eta)}P_{\theta}(\mu_{1}(\theta)\in\{\rho(W_{n})\})=1-o(1)$.
Since $\Theta_{binary}(0)\subset\Theta_{binary}(\eta)$, it follows
that 
\[
\liminf_{n\rightarrow\infty}\inf_{\theta\in\Theta_{binary}(0)}P_{\theta}(\mu_{1}(\theta)\in\{\rho(W_{n})\})=1-o(1).
\]

However, this contradicts $\liminf_{n\rightarrow\infty}\inf_{\theta\in\Theta_{binary}(0)}P_{\theta}(\{-1,0\}\subset\{\rho(W_{n})\})\geq1-\alpha$
since $\{\rho(W_{n})\}$ is a singleton. The proof is complete. 
\end{proof}

\subsection{Proof of Theorem \ref{thm: phase trans binary hard}}
\begin{lem}
\label{lem: binary lem experiment part 1}Let $\theta=(a,b,c,H)$
satisfy $P_{\theta}(Y_{i}\in\{0,1\})=1$. Assume that $P_{\theta}(Y_{i}=D_{i}=0\mid Z_{i}=1)$
and $k_{1}-k_{2}$ are bounded below by a positive constant. If $c=0$,
$k_{2}=o(1)$ and $\beta=o(1)$ satisfy $\beta<0$ and $\beta(k_{1}-k_{2})+P_{\theta}(Y_{i}=1,D_{i}=1\mid Z_{i}=0)\geq0$,
then for large enough $n$ there exists $\ttheta=(\ta,\tb,\tc,\tH)$
such that (1) $P_{\theta}$ and $P_{\ttheta}$ imply the same distribution
for the observed data $(Y_{i},D_{i},Z_{i})$ and (2) $\mu_{1}(\ttheta)\geq0$
and and $\tc\leq k_{2}$.
\end{lem}
\begin{proof}
We repeat all the arguments in the proof of Lemma \ref{lem: binary lem part 1}
including Equation (\ref{eq: binary lem part 1 eq 8}), which we repeat
here:
\[
\beta(k_{1}-k_{2})+\min\{\tc,\rho_{1,0}k_{2}\}-\max\{0,\tc+(t_{0,0}-1)(1-k_{1})\}\geq0.
\]

We inherit all the notations from the proof of Lemma \ref{lem: binary lem part 1}.
The only difference is that $\tc=\rho_{1,0}k_{2}$ (rather than $\tc=\min\{\eta,k_{2}\}$).
As argued in (\ref{eq: binary lem part 1 eq 2.5}) in the proof of
Lemma \ref{lem: binary lem part 1}, $(1-t_{0,0})(1-k_{1})=P_{\theta}(Y_{i}=D_{i}=0\mid Z_{i}=1)$,
which by assumption is bounded below by a positive constant. Since
$\tc\leq k_{2}\rightarrow0$, the above display for large enough $n$
becomes
\begin{equation}
\beta(k_{1}-k_{2})+\min\{\tc,\rho_{1,0}k_{2}\}\geq0.\label{eq: binary experiment lem part 1 eq 3}
\end{equation}

By $\tc=\rho_{1,0}k_{2}$, this becomes $\beta(k_{1}-k_{2})+\rho_{1,0}k_{2}\geq0$.
As argued in (\ref{eq: binary lem part 1 eq 2}) in the proof of Lemma
\ref{lem: binary lem part 1}, $\rho_{1,0}k_{2}=P_{\theta}(Y_{i}=1,D_{i}=1\mid Z_{i}=0)$.
Thus, $\beta(k_{1}-k_{2})+\rho_{1,0}k_{2}\geq0$ follows by assumption.
The proof is complete. 
\end{proof}
\begin{lem}
\label{lem: binary lem experiment part 2}Let $\theta=(a,b,c,H)$
satisfy $P_{\theta}(Y_{i}\in\{0,1\})=1$. Assume that $P_{\theta}(Y_{i}=D_{i}=0\mid Z_{i}=1)$
and $k_{1}-k_{2}$ are bounded below by a positive constant. If $c\leq k_{2}=o(1)$
and $\beta=o(1)$ satisfy $\beta<0$ and $\beta(k_{1}-k_{2})+P_{\theta}(Y_{i}=1,D_{i}=1\mid Z_{i}=0)<0$,
then $\mu_{1}(\theta)<0$. 
\end{lem}
\begin{proof}
Recall that $\beta=\frac{\mu_{1}(\theta)\cdot b-\mu_{2}(\theta)\cdot c}{b-c}$,
which means that $\mu_{1}(\theta)=\frac{c}{b}\mu_{2}(\theta)+\left(1-\frac{c}{b}\right)\beta$.
Since $\beta<0$, it suffices to show that $c\mu_{2}(\theta)+(b-c)\beta<0$.
Since $b-c=k_{1}-k_{2}$, we need to show that 
\begin{equation}
c\mu_{2}(\theta)<-\beta(k_{1}-k_{2}).\label{eq: binary experiemnt lem part 2 eq 2}
\end{equation}

Following a similar computation as in the proof of Lemma \ref{lem: binary lem part 1},
we have 
\[
E_{\theta}(Y_{i}\mid D_{i}=1,Z_{i}=0)=r_{1,1}\frac{k_{2}-c}{k_{2}}+r_{0,1}\frac{c}{k_{2}},
\]
where $r_{1,1}=E_{\theta}(Y_{i}(1)\mid D_{i}(1)=1,D_{i}(0)=1)$ and
$r_{0,1}=E_{\theta}(Y_{i}(1)\mid D_{i}(1)=0,D_{i}(0)=1)$. Since $P_{\theta}(D_{i}=1\mid Z_{i}=0)=k_{2}$,
we have
\begin{multline*}
P_{\theta}(Y_{i}=1,D_{i}=1\mid Z_{i}=0)=P_{\theta}(Y_{i}=1\mid D_{i}=1,Z_{i}=0)\cdot P_{\theta}(D_{i}=1\mid Z_{i}=0)\\
=r_{1,1}(k_{2}-c)+r_{0,1}c\geq r_{0,1}c\overset{\texti}{\geq}\mu_{2}(\theta)c,
\end{multline*}
where (i) follows by the fact that $\mu_{2}(\theta)=r_{0,1}-E_{\theta}(Y_{i}(0)\mid D_{i}(1)=0,D_{i}(0)=1)\leq r_{0,1}$.
By the assumption of $\beta(k_{1}-k_{2})+P_{\theta}(Y_{i}=1,D_{i}=1\mid Z_{i}=0)<0$,
we have 
\[
\beta(k_{1}-k_{2})+c\mu_{2}(\theta)\leq\beta(k_{1}-k_{2})+P_{\theta}(Y_{i}=1,D_{i}=1\mid Z_{i}=0)<0.
\]

This proves (\ref{eq: binary experiemnt lem part 2 eq 2}). The proof
is complete.
\end{proof}
\begin{proof}[\textbf{Proof of Theorem \ref{thm: phase trans binary hard}}]
The proof follows the same argument as the proof of Theorem \ref{thm: phase trans binary easy},
except that Lemmas \ref{lem: binary lem part 1} and \ref{lem: binary lem part 2}
is replaced by Lemmas \ref{lem: binary lem experiment part 1} and
\ref{lem: binary lem experiment part 2}.
\end{proof}

\subsection{Proof of Theorems \ref{thm: key} and \ref{thm: new ID sign}}
\begin{proof}[\textbf{Proof of Theorem \ref{thm: key}}]
Without loss of generality, we assume $cov(D_{i},Z_{i})>0$. (This
is because we can swap the 0-1 labels for $Z_{i}$ to obtain $cov(D_{i},Z_{i})>0$
and such a swap does not change $\max\{|\mu_{1}|,|\mu_{2}|\}$.)

The case for $cov(D_{i},Z_{i})<0$ follows by swap $Z_{i}$ and $1-Z_{i}$.
We recall the notation of $k_{1}=E(D_{i}\mid Z_{i}=1)=E(D_{i}(1))=a+b$
and $k_{2}=E(D_{i}\mid Z_{i}=0)=E(D_{i}(0))=a+c$. We observe that
$cov(D_{i},Z_{i})>0$ implies that $E(D_{i}\mid Z_{i}=1)>E(D_{i}\mid Z_{i}=0)$,
which means $k_{1}-k_{2}=b-c>0$. We notice that $\max\{|\mu_{1}|,|\mu_{2}|\}=\max\{\mu_{1},-\mu_{1},\mu_{2},-\mu_{2}\}$.
By $\beta=(\mu_{1}b-\mu_{2}c)/(b-c)$, we have
\[
\mu_{1}=\lambda\mu_{2}+(1-\lambda)\beta,
\]
where $\lambda=c/b$. This means that $\max\{|\mu_{1}|,|\mu_{2}|\}=f(\mu_{2},\lambda)$,
where 
\[
f(\mu_{2},\lambda):=\max\left\{ \lambda\mu_{2}+(1-\lambda)\beta,\ -\lambda\mu_{2}-(1-\lambda)\beta,\ \mu_{2},\ -\mu_{2}\right\} .
\]

Recall that $k_{1}=a+b$, $k_{2}=a+c$ and $b-c=k_{1}-k_{2}>0$. Thus,
$\lambda=c/(k_{1}-k_{2}+c)$ with $c\in[0,k_{2}]$. This means that
$\lambda\in[0,k_{2}/k_{1}]\subset[0,1)$. We now find 
\[
\min_{\mu_{2}\in\RR,\lambda\in[0,k_{2}/k_{1}]}f(\mu_{2},\lambda).
\]

For any $\lambda>0$, we notice that $\min_{\mu_{2}\in\RR}f(\mu_{2},\lambda)$
must occur at a point such that two of the four components of $f(\mu_{2},\lambda)$
are equal; otherwise, we can change $\mu_{2}$ slightly to lower the
largest component even further to lower $f(\mu_{2},\lambda)$. Therefore,
$\mu_{2}(\lambda)\in\arg\min_{\mu_{2}\in\RR}f(\mu_{2},\lambda)$ if
and only if at least one of the following is true:
\begin{itemize}
\item $\lambda\mu_{2}(\lambda)+(1-\lambda)\beta=\mu_{2}(\lambda)$ (i.e.,
$\mu_{2}(\lambda)=\beta$)
\item $\lambda\mu_{2}(\lambda)+(1-\lambda)\beta=-\mu_{2}(\lambda)$ (i.e.,
$\mu_{2}(\lambda)=-\beta(1-\lambda)/(1+\lambda)$)
\item $\mu_{2}(\lambda)=-\mu_{2}(\lambda)$ (i.e., $\mu_{2}(\lambda)=0$)
\item $\lambda\mu_{2}(\lambda)+(1-\lambda)\beta=-\lambda\mu_{2}(\lambda)-(1-\lambda)\beta$
(i.e., $\mu_{2}(\lambda)=(1-\lambda^{-1})\beta$).
\end{itemize}
We plug these four values of $\mu_{2}(\lambda)$ into $f(\mu_{2},\lambda)$
and take the minimum of the four values of $f(\mu_{2},\lambda)$,
obtaining that for $\lambda>0$, 
\[
\min_{\mu_{2}\in\RR}f(\mu_{2},\lambda)=|\beta|\cdot\min\left\{ 1,\left|1-\lambda\right|,\frac{\left|1-\lambda\right|}{1+\lambda},\frac{\left|1-\lambda\right|}{\lambda}\right\} .
\]

Since $\lambda\in[0,1)$, it follows that for $\lambda>0$, 
\[
\min_{\mu_{2}\in\RR}f(\mu_{2},\lambda)=|\beta|\cdot\frac{\left|1-\lambda\right|}{1+\lambda}.
\]

Now we take the infimum over $\lambda\in(0,k_{2}/k_{1}]$, obtaining
\[
\inf_{\mu_{2}\in\RR,\lambda(0,k_{2}/k_{1}]}f(\mu_{2},\lambda)=\inf_{\lambda\in(0,k_{2}/k_{1}]}|\beta|\cdot\frac{\left|1-\lambda\right|}{1+\lambda}=|\beta|\cdot\frac{k_{1}-k_{2}}{k_{1}+k_{2}}.
\]

We observe that $\min_{\mu_{2}\in\RR}f(\mu_{2},0)=|\beta|$. Thus,
\[
\min_{\mu_{2}\in\RR,\lambda\in[0,k_{2}/k_{1}]}f(\mu_{2},\lambda)=|\beta|\cdot\frac{k_{1}-k_{2}}{k_{1}+k_{2}}.
\]

Therefore, we have proved that for any $\mu_{2}\in\RR$ and for any
$\lambda\in[0,k_{2}/k_{1}]$,
\[
\max\{|\mu_{1}|,|\mu_{2}|\}\geq|\beta|\cdot\frac{k_{1}-k_{2}}{k_{1}+k_{2}}.
\]

The proof is complete. 
\end{proof}

\begin{proof}[\textbf{Proof of Theorem \ref{thm: new ID sign}}]
We recall that $\beta=(\mu_{1}b-\mu_{2}c)/(b-c)$, which means 
\begin{equation}
\mu_{1}=(c/b)\mu_{2}+(1-c/b)\beta.\label{eq: thm new ID sign 3}
\end{equation}

We observe that $cov(D_{i},Z_{i})>0$ implies that $E(D_{i}\mid Z_{i}=1)>E(D_{i}\mid Z_{i}=0)$,
which means $b-c>0$. We consider two cases: (A) $\beta>0$ and (B)
$\beta<0$. 

\textbf{Step 1:} prove the result in the case of $\beta>0$.

We proceed by contradiction. Suppose that $|\mu_{1}|\geq|\mu_{2}|$
and $\mu_{1}\leq0$. By (\ref{eq: thm new ID sign 3}), this means
$c\mu_{2}+(b-c)\beta\leq0$, which can be written as $\mu_{2}\leq-(b-c)\beta/c$.
Since $b-c>0$ and $\beta>0$, this means that $\mu_{2}\leq0$. Since
$\mu_{1},\mu_{2}\leq0$ and $|\mu_{1}|\geq|\mu_{2}|$, we have $\mu_{1}\leq\mu_{2}$.
By (\ref{eq: thm new ID sign 3}), this means 
\[
(c/b)\mu_{2}+(1-c/b)\beta\leq\mu_{2}.
\]

Using $b-c$, we obtain $\mu_{2}\geq\beta$. Since $\beta>0$, we
have $\mu_{2}>0$. However, this contradicts $\mu_{2}\leq0$. Therefore,
we have proved $\mu_{1}>0$. 

\textbf{Step 2:} prove the result in the case of $\beta<0$.

The argument is analogous. We state the argument here for completeness.
Suppose that $|\mu_{1}|\geq|\mu_{2}|$ and $\mu_{1}\geq0$. By (\ref{eq: thm new ID sign 3}),
this means $c\mu_{2}+(b-c)\beta\geq0$, which can be written as $\mu_{2}\geq-(b-c)\beta/c$.
Since $b-c>0$ and $\beta<0$, this means that $\mu_{2}\geq0$. Since
$\mu_{1},\mu_{2}\geq0$ and $|\mu_{1}|\geq|\mu_{2}|$, we have $\mu_{1}\geq\mu_{2}$.
By (\ref{eq: thm new ID sign 3}), this means 
\[
(c/b)\mu_{2}+(1-c/b)\beta\geq\mu_{2}.
\]

Using $b-c$, we obtain $\mu_{2}\leq\beta$. Since $\beta<0$, we
have $\mu_{2}<0$. However, this contradicts $\mu_{2}\geq0$. Therefore,
we have proved $\mu_{1}<0$. 

Therefore, we have proved that in both cases, $\sign(\mu_{1})=\sign(\beta)$. 
\end{proof}

\end{document}